\documentclass[journal]{IEEEtran}
\usepackage{graphicx}
\usepackage[cmex10]{amsmath}
\usepackage{array}
\usepackage{amssymb}
\usepackage{amsthm}
\usepackage{mathtools,bm}
\usepackage{mathrsfs}
\usepackage{bbm} 				
\usepackage{xcolor}
\usepackage{cite}
\usepackage{siunitx}
\usepackage{stackrel}
\usepackage{wrapfig}
\usepackage{pgfplots}
\usetikzlibrary{automata, positioning}
\usepackage{csvsimple}
\usepackage{arydshln}
\usepackage{scalerel}
\usepackage{lipsum}
\usepackage[ruled,vlined,linesnumbered]{algorithm2e}
\usepackage{hyperref}
\hypersetup{
    colorlinks=true,
    linkcolor=blue,
    citecolor=blue,
    filecolor=magenta,      
    urlcolor=cyan,
    pdftitle={Overleaf Example},
    pdfpagemode=FullScreen,
}
\usepackage{subfigure}
\usepackage{upgreek}
\usepackage{todonotes}

\newtheorem{corollary}{Corollary}
\makeatletter
\def\BState{\State\hskip-\ALG@thistlm}
\makeatother
\usepackage{colortbl}
\usepgfplotslibrary{fillbetween}
\pgfplotsset{every tick label/.append style={font=\scriptsize}}
\pgfplotsset{every axis label/.append style={font=\scriptsize}}
\pgfplotsset{legend style={font=\scriptsize}}
\pgfplotsset{
	compat=1.11,
	legend image code/.code={
		\draw[mark repeat=2,mark phase=2]
		plot coordinates {
			(0cm,0cm)
			(0.25cm,0cm)        
			(0.5cm,0cm)         
		};%
	}
}
\DeclarePairedDelimiter\abs{\lvert}{\rvert}%

\newcommand{\qtilde}[0]{\tilde{\bm{q}}}
\newcommand{\dbm}[0]{\bm{d}}

\newcommand{\XX}{\mbox{\tiny \it XX'}}

\newcommand{\X}{\mbox{\tiny \it X}}

\theoremstyle{plain}
\newtheorem{thm}{Theorem}
\newtheorem{lem}{Lemma} 
\newtheorem{prop}{Proposition}

\newtheorem{rem}{Remark}
\newtheorem{defn}{Definition}

\newtheorem{exmp}{Example}
\DeclareRobustCommand{\T}{^\intercal}
\theoremstyle{remark}

\usepackage{graphicx}
\usepackage[cmex10]{amsmath}
\usepackage{array}
\usepackage{amssymb}
\usepackage{amsthm}
\usepackage{mathtools,bm}
\usepackage{mathrsfs}
\usepackage{bbm} 				
\usepackage{xcolor}
\usepackage{cite}
\usepackage{siunitx}
\usepackage{stackrel}
\usepackage{wrapfig}
\usepackage{pgfplots}
\usetikzlibrary{automata, positioning}
\usepackage{csvsimple}
\usepackage{arydshln}
\usepackage{scalerel}
\usepackage{lipsum}
\usepackage[ruled,vlined,linesnumbered]{algorithm2e}
\usepackage{hyperref}
\hypersetup{
    colorlinks=true,
    linkcolor=blue,
    citecolor=blue,
    filecolor=magenta,      
    urlcolor=cyan,
    pdftitle={Overleaf Example},
    pdfpagemode=FullScreen,
}
\usepackage{subfigure}
\usepackage{upgreek}
\usepackage{todonotes}
\makeatletter
\def\BState{\State\hskip-\ALG@thistlm}
\makeatother
\usepackage{colortbl}
\usepgfplotslibrary{fillbetween}
\pgfplotsset{every tick label/.append style={font=\scriptsize}}
\pgfplotsset{every axis label/.append style={font=\scriptsize}}
\pgfplotsset{legend style={font=\scriptsize}}
\pgfplotsset{
	compat=1.11,
	legend image code/.code={
		\draw[mark repeat=2,mark phase=2]
		plot coordinates {
			(0cm,0cm)
			(0.25cm,0cm)        
			(0.5cm,0cm)         
		};%
	}
}
\newcommand{\response}[1]{\textcolor{black}{#1}}
\usepackage[]{cleveref}
\crefname{lemma}{lemma}{lemmas}
\Crefname{lemma}{Lemma}{Lemmas}
\crefname{thm}{theorem}{theorems}
\Crefname{thm}{Theorem}{Theorems}
\crefname{defn}{definition}{definitions}
\Crefname{defn}{Definition}{Definitions}
\crefname{prop}{proposition}{propositions}
\Crefname{prop}{Proposition}{Propositions}
\crefname{lem}{lemma}{lemmas}
\Crefname{lem}{Lemma}{Lemmas}
\crefname{algorithm}{algorithm}{algorithms}
\Crefname{algorithm}{Algorithm}{Algorithms}
\crefname{subsection}{subsection}{subsections}

\title{Differential Privacy for Class-based Data:\\A Practical Gaussian Mechanism}
\author{ Raksha Ramakrishna,~\IEEEmembership{Member, ~IEEE}, Anna Scaglione,~\IEEEmembership{Fellow, ~IEEE}, Tong Wu,~\IEEEmembership{Member,~IEEE}, \\
Nikhil Ravi,~\IEEEmembership{Student Member, ~IEEE} and Sean Peisert,~\IEEEmembership{Senior Member,~ IEEE}%
\thanks{Raksha Ramakrishna is with Division of Network and Systems Engineering, EECS, KTH Royal Institute of Technology, Stockholm, Sweden. e-mail: rakshar@kth.se.}
\thanks{Anna Scaglione, Tong Wu and Nikhil Ravi are with the Department of Electrical and Computer Engineering, Cornell Tech, New York City, NY 10044  USA. e-mail:\{as337, tw385, nr337\}@cornell.edu.}
\thanks{Sean Peisert is with the Computer Sciences Research and Development
Department, Lawrence Berkeley National Laboratory, Berkeley, CA 94720
USA. e-mail: sppeisert@lbl.gov}
\thanks{This research was supported by the Director, Cybersecurity, Energy Security, and Emergency Response, Cybersecurity for Energy Delivery Systems program, of the U.S. Department of Energy, under contract DE-AC02-05CH11231.  Any opinions, findings, conclusions, or recommendations expressed in this material are those of the authors and do not necessarily reflect those of the sponsors of this work.}
}
\begin{document}
	\maketitle 
	
\begin{abstract}
    In this paper, we present a notion of differential privacy (DP) for data that comes from different classes. Here, the class-membership is private information that needs to be protected. The proposed method is an output perturbation mechanism that adds noise to the release of query response such that the analyst is unable to infer the underlying class-label. The proposed DP method is capable of not only protecting the privacy of class-based data but also meets quality metrics of accuracy and is computationally efficient and practical. We illustrate the efficacy of the proposed method empirically while outperforming the baseline \response{additive} Gaussian noise mechanism. We also examine a real-world application and apply the proposed DP method to the autoregression and moving average (ARMA) forecasting method, protecting the privacy of the underlying data source. 
    Case studies on the real-world advanced metering infrastructure (AMI) measurements \response{of household power consumption} validate the excellent performance of the proposed DP method while also satisfying the accuracy of forecasted \response{power consumption} measurements.
\end{abstract}	

\section{Introduction}
Differential privacy (DP)\cite{dwork2008differential,dwork2014algorithmic}  has become one of the most important concepts in database privacy, gaining an important foothold in ensuring that personal or data private to individual database entries remains indistinguishable  after the database is queried. Broadly speaking, a privacy mechanism for a certain data query is differentially private if the output of this data query changes minimally, or with a very low probability when  a single database entry is added or removed. 
Numerous papers give a comprehensive overview of differential privacy. In \cite{sarwate2013signal}, the authors treat the problem of differential privacy from a signal processing perspective and provide a review of the different mechanisms employed. Differential privacy mechanisms can be broadly classified as input perturbation, i.e., adding noise to the input or to the data before responding to the query, and output perturbation, noise addition after the query response is computed. 
Most of the time, Laplacian or Gaussian noise is added to the input or the output of the query. The so-called {\it exponential mechanism}~\cite{mcsherry2007mechanism,nissim2012approximately} is another often-used differential privacy mechanism with good guarantees, particularly in the case where the utility of differentially private query response is also taken into account. 
\par Traditional or classical Gaussian mechanisms \cite{dwork2008differential} add uncorrelated Gaussian noise, whose variance is calculated on the basis of the sensitivity of neighboring sets of private data to the query and privacy-level parameters. It was shown in \cite{balle2018improving} that the variance could be further reduced given a privacy budget to obtain better utility. A post-processing or denoising mechanism is also introduced to improve the accuracy of the query.

\par Despite the popularity and use of DP in a variety of situations, there are cases where it is not possible to employ the original definition and mechanisms for DP. A main drawback with the traditional DP is that the definition of neighboring datasets which pertains to the addition or removal of a single data point is not applicable to many use cases. For example, query response for a single data stream where some important attributes of this data stream are to be protected rather than hiding the presence or absence of a single data point in a large database or data stream. 
To address this drawback, there have been other notions of privacy that are more applicable to scenarios that do not necessarily cater to large databases. Two such frameworks are the Pufferfish~\cite{KiferPufferFish} and Blowfish~\cite{he2014blowfish} privacy. Here, they define privacy with respect to pairs of `secrets' that must remain indistinguishable after the privacy mechanism is used to release the data. 
These frameworks have been successfully utilized in many scenarios and particularly that of trajectory or location release of an individual \cite{xiao2015protecting, xiao2016dphmm, xiao2017loclok} or activity monitoring \cite{song2017pufferfish}. However, the privacy mechanisms presented in the aforementioned papers are either computationally expensive, do not scale with the size of sample space, or do not incorporate the utility of the released data while designing the mechanism. Furthermore, none of the papers mentioned  discuss an output perturbation method such as the Gaussian mechanism, which optimizes the accuracy of the query.  

\par As an application, we consider the release of forecasts performed fitting an Auto-Regressive Moving-Average (ARMA) model. ARMA prediction models have many variants, which are widely used for time-series forecasts; the underlying assumption is that the time series is a realization of a Gaussian process, with a parametric structure for its covariance that is determined by the parameters of an ARMA filter. The Gaussian assumption implies that the Minimum Mean Squared Error (MMSE) forecast is the conditional mean of the future samples given the past, and it is a linear affine function of the observation that depends on the first and second-order statistics of the process.  A useful application we explore is the release of the predicted power consumption for homes in a way that its statistics are $(\epsilon,\delta)$ private relative to other homes which could be, for instance, in the same neighborhood.  
Releasing time-series in a private manner is yet another paradigm where the definition of neighboring data or adjacency is different, as it is in \cite{fioretto2019optstream,le2020differential}. The query response mechanisms in these papers correspond to down-sampling data points in a window, adding DP-noise, and reconstructing the time-series which is shown to be differentially private.

\par An alternative way to make forecasts private would be to carry out input perturbation or regression in a differentially private manner. Recent work \cite{sheffet2019old} studies linear regression in a differentially private manner, where the input data and the corresponding labels are made private by perturbing sufficient statistics. Related to the proposed Gaussian mechanism, in \cite{bernstein2019differentially}, a Bayesian linear regression is carried out, i.e., the posterior of the regression parameters is computed in a private manner, to protect the underlying data, by perturbing the terms comprised of private data that are required for estimation of regression parameters. Furthermore, the  noisy posterior is `denoised' or, in other words, a noise-aware inference is undertaken.  
In \cite{honkela2021gaussian}, both input data and labels are made private by the addition of noise to sufficient statistics when the query is  Gaussian process regression. Like in \cite{bernstein2019differentially}, noise-aware posterior computation of parameters is discussed.  
Compared to the papers discussed above, the differences not only include a change to the definition of a neighborhood of data but also the output perturbation mechanism, whose privacy guarantees are easier to analyze. 

\subsection{Differential privacy in smart grid systems}
 \response{Utilities and distributed energy resource (DER) providers are increasingly enhancing their data collection capabilities and implementing data-driven digital transformations in their organizations~\cite{currie2023data}. Alongside this trend, there is also a rising demand from third parties to access this data. For instance, organizations seek to identify optimal locations for DER sites, while commercial industries aim to leverage this data to drive digital economy breakthroughs in the energy sector. In addition, law enforcement agencies investigate cybercrimes on these infrastructures, while regulatory bodies work towards decarbonization and grid modernization goals~\cite{unitedstatesdepartmentofenergyofficeofelectricity2020strategy}. However, the growth in the number of measurements collected in distribution systems has led to exacerbated data privacy concerns~\cite{liu2012cyber}.} 
 As mentioned before, the practical application of our DP mechanism we highlight is sharing power consumption forecasts.
Differential privacy mechanisms have been widely used in class-based smart grid data classification and forecasting. For example, \cite{wang2020privacy}   injects functional DP noises into the meter data to achieve a certain level of differential privacy. In \cite{random_BLH}, the authors proposed  a novel randomized battery-based load-hiding algorithm that assures differential privacy for smart metering data. In \cite{gough2021preserving}, an innovative DP compliant algorithm was developed to ensure that the data from consumer’s smart meters are protected. Moreover,  spectral DP is presented to protect the frequency content of  power system time-series data in \cite{parker2021spectral, ou2020singular}. DP techniques are also applied for power system operation to protect users' energy consumption patterns in \cite{fioretto2019differential, dvorkin2020differentially, dvorkin2020differentially2}. In particular, \cite{dvorkin2020differentially, dvorkin2020differentially2} consider a privacy-preserving optimal power flow (OPF) mechanism for distribution grids that secures customer privacy from unauthorized access to OPF solutions, e.g., current and voltage measurements. In \cite{fioretto2019differential}, the authors investigate  how to utilize DP techniques to release the data for power networks where the parameters of transmission lines and transformers are obfuscated.  None of these works consider the release of forecasts that are DP.

\subsection{Contributions}
We introduce the concept of differential privacy for \response{class-based} data. Our framework is related to the Pufferfish privacy mechanism, as discussed above. The sensitive or private information here is the class-label or the hypothesis, rather than the presence or absence of a single data sample.  We propose an additive noise mechanism for the release of the query response on such data so that the analyst is unable to infer the underlying class-label.
The optimality of the \response{query response mechanism} rests on the assumption that the data from each class have a Gaussian multivariate distribution with a class-specific mean and covariance.  Furthermore, we propose a Gaussian query response mechanism that is computationally efficient and practical because it also meets accuracy requirements. We apply the DP mechanism for the release of ARMA forecasts, testing it numerically on synthetic and real data. \response{In particular we,
\begin{itemize}
    \item Define the concept of differential privacy for class-based data by considering an appropriate notion of the neighborhood for a class-label. The notion is borrowed from Pufferfish privacy framework.
    \item Propose an additive noise query response mechanism for class-based data that is optimal when the query is conditionally Gaussian given class-label. The mechanism adds Gaussian noise to the query where the noise parameters are optimally chosen such that a certain level of utility is met while achieving privacy guarantees that are better than the addition of white Gaussian noise.
    \item Illustrate the application of the proposed mechanism when the query is ARMA forecast given a class-label. 
    \item Discuss the real-world application of power consumption forecasts for households while maintaining household-level privacy and satisfying utility requirements. 
\end{itemize}}

\subsection{\response{Organization}}
\response{In Section \ref{sec:prob_setting}, we outline the problem setting and define neighborhood in the context of class-based differential privacy. Then, in Section \ref{sec:metrics}, we review relevant concepts from the literature. In Section \ref{sec:publishing_mechanism}, we introduce additive-noise based publishing mechanism to achieve differential privacy in class-based data. In Section \ref{sec:ARMA_forecasts}, we discuss the application of class-labeled ARMA forecasts since it serves as the theoretical background for the real-world example of forecasting household electricity consumption data.   Section \ref{sec:num_results}, we discuss numerical results for synthetic data and household electricity consumption forecast data. Finally in Section \ref{sec:conclusions}, conclusions and future work are detailed.}

\subsection{Notation}
Boldfaced lowercase letters, $\bm{x}$, are used to denote vectors whereas upper-case letters, $\bm{X}$, are for matrices. 
Calligraphic letters, $\mathcal{X}$, are used to denote sets. \response{We use the symbol $\bm \Sigma$ to denote covariance matrices and $\bm \mu$ to denote mean vectors, and use the suffix to indicate what is the random vector that is averaged; $\mathrm{Tr}(\bm A)$ denotes the trace of the matrix $\bm A$.}

\section{Problem setting}\label{sec:prob_setting}
In this section, we elucidate the problem and review the appropriate definitions from the literature. The setting of the problem is that the \textit{data owner} requires sensitive information to be private, while also wanting to answer queries by a third-party analyst with as little error or distortion to the query response as possible. 

In this work, we denote by $X \in \mathcal{X}$ the sensitive information that the \textit{data owner} wants to hold private.  One can also refer to this as the label, class, or hypothesis $X$, that is hidden from the analyst. For brevity, we use \textit{label} to refer to the sensitive information $X$.
Let the probability distribution of data, $\bm{d} \in \mathbb{R}^{n}$, that is assumed to be generated given the label, $X$, be $f(\bm{d}|X)$.  
 
The {\it analyst} wishes to apply queries or functions $\mathbbm{Q}: \mathbb{R}^{n} \rightarrow \mathbb{R}^{k} $ to the data $\dbm$. We denote the outcome of the query as $\bm{q} \in \mathbb{R}^{k}$, i.e., $\bm{q}=\mathbb{Q}(\dbm)$ which is a stochastic function or function of a random variable. 
We assume that the query at a certain instance only acts on data generated when the system is in a specific label $X$, i.e., the query does not combine information from multiple instances with the system in different labels. Thus, we denote the probability distribution of $\bm{q}$  as $f(\bm{q}|X)$. Note that the query output is already modeled as a \textit{random outcome}. However, releasing $\bm{q}$ as is could aid a malicious analyst in inferring the label $X$ through standard classification schemes with prior information or belief about $X$.
Therefore, rather than answering the query directly $\bm{q}=\mathbb{Q}(\dbm)$, the  \textit{data owner} publishes the output of a randomized algorithm or mechanism denoted by $\mathbbm{A}_{\mathrm Q}$. Thus, the \textit{data owner} releases or publishes query response,
$\qtilde=\mathbbm{A}_{\mathrm Q}(\dbm|X)$. Again, due to the dependence on a label $X$, we denote the probability density of the released query as $f(\qtilde|X)$. 

The goal of the randomized mechanism $\mathbbm{A}_{\mathrm Q}$ is to confuse the analyst by making the query output given the label $X$ indistinguishable from the query output when the underlying label is a \textit{neighbor} to $X$. The definition of neighboring labels will be made clearer later.  

For the published data to be useful, the answer to the query $\qtilde$ is required to meet certain quality metrics, such as accuracy (denoted by $\rho$). For the rest of the paper, \textit{accuracy} refers to accuracy in expectation over the domain of the privatized data.

\subsection{\response{Neighborhood of a class-label}} 
The \textit{data owner} needs to define, a priori for every label $X$, a subset of labels $X'\in {\cal X}$ as neighbors. These neighbor labels are chosen such that privacy is maintained based on our class-membership-based definition. In other words, the labels $X$ and $ X'$ should be indistinguishable from the query answer alone. This is similar to the concept of pair of secrets in the Pufferfish privacy framework \cite{song2017pufferfish}. 

A graph topology ${\cal G}$ whose nodes' set is ${\cal X}$ and the edges set, ${\cal E}$, can encode information on what ought to be hidden,
\begin{align}
 {\cal G}=({\cal X},{\cal E}), ~~{\cal E}=\{(X,X')|X\in {\cal X},  X' \in {\cal X}_X^{(1)}\},
\end{align}
One could use a notion of distance $\mathbbm{d}(X,X')$ that can define a geometric graph structure where the subset ${\cal X}^{(1)}_{X}$ that is the neighborhood ${\cal X}_X^{(1)}$ of each node $X$ is defined as:
\begin{align}
    {\cal X}^{(1)}_{X}=\{X'|\mathbbm{d}(X,X')=1,X'\in{\cal X}\}, 
\end{align}
describing all the $X'$ that are {\it at distance one} from $X$
 and should be giving similar query answers as $X$. Furthermore, the graph structure is assumed to be undirected, 
 \begin{align}
   X' \in {\cal X}_X^{(1)} \Leftrightarrow X \in {\cal X}_{X'}^{(1)}   
 \end{align}
 The \textit{data owner} needs to design the neighborhood graph ${\cal G}$ that determines the neighborhood for each label $X$. Neighborhoods can be designed based on desired indistinguishable  labels. In the case of very strict privacy requirements, a fully connected graph can be used as a neighborhood graph so that every label has all other labels as neighbors.

Next, we review the relevant concepts from the literature to compare and contrast the ideas from the problem setting described above.

\section{\response{Review of Privacy Definitions and Metrics}}\label{sec:metrics}

\subsection{Definition of Privacy }
Conventionally, given the random published answer or realization of $\qtilde$ in the differential privacy literature, the  name {\it privacy loss} is used as a synonym for the log-likelihood ratio:  
\begin{align}
L_{\XX}(\qtilde)&\triangleq\ln \frac{f(\qtilde|X)}{f(\qtilde|X')}. \label{eq:privacycosnt}
\end{align}
The reason for the name is that in classical statistical inference, $L_{\XX}(\qtilde)>0$ yields the decision that the query output is generated by the distribution $\bm{f}(\qtilde|X)$. If this event is infrequent, then often an alternative hypothesis $X'$ (where the distribution is $\bm{f}(\qtilde|X')$) will be chosen as the right probabilistic model. 
We now introduce the notion of $(\epsilon, \delta)$ differential privacy, which applies to any random vector $\qtilde$ that is not conditionally independent of the private information $X$:

\begin{defn}[$(\epsilon, \delta)$ Probabilistic Differential Privacy (PDP) \cite{Machanavajjhala}]\label[defn]{defn:e-d-privacy} Consider the probability density of released query or randomized mechanism $\qtilde \sim f(\qtilde|X)$ that changes depending on the class $X\in {\cal X}$. The randomized mechanism producing $\qtilde$ is  $(\epsilon, \delta)$- Probabilistic Differentially Private (PDP)  iff: 
\begin{align}
Pr\left( \abs{L_{\XX}
    (\qtilde)}>\epsilon\right) \leq \delta ~~ \forall (X, X') \in {\cal E} \label{eq:epsilon_delta_relation}.
\end{align}
\end{defn}
It can be shown that $(\epsilon,\delta)$-PDP is a strictly stronger condition than $(\epsilon,\delta)$-DP. 
\begin{thm}[PDP implies DP~\cite{mcclure2015relaxations}]\label[thm]{thm:PDP-DP}
If a randomized mechanism is $(\epsilon,\delta)$-PDP, then it is also $(\epsilon,\delta)$-DP, i.e.,
\[
    (\epsilon,\delta)-\text{PDP}  \Rightarrow (\epsilon,\delta)-\text{DP}, \text{ but } (\epsilon,\delta)-\text{DP} \nRightarrow (\epsilon,\delta)-\text{PDP}.
\]
\end{thm}

Given that PDP provides a more intuitive understanding of privacy than $(\epsilon,\delta)$ DP and is a strictly stronger condition, we make use of PDP throughout this paper. 
 PDP is not closed under post-processing \cite{meiser2018approximate} only if, prior to the query, one applies a non-bijective transformation.
Going forward, we drop the absolute value while writing $\abs{L_{\XX}
    (\qtilde)}>\epsilon$ since 
    \begin{align}
    \abs{L_{\XX}
    (\qtilde)}>\epsilon \implies L_{\XX}
    (\qtilde) > \epsilon , L_{\XX}
    (\qtilde) < -\epsilon \\
     L_{\XX}
    (\qtilde) < -\epsilon  \implies  -L_{\XX}
    (\qtilde) > \epsilon  \implies L_{\XX} (\qtilde) > \epsilon 
    \end{align}
As the neighborhood graph is undirected, and $Pr\left( \abs{L_{\XX}
    (\qtilde)}>\epsilon\right) \leq \delta ~~ \forall (X, X') \in {\cal E}$, it suffices to drop the absolute value and 
    $Pr\left( {L_{\XX}
    (\qtilde)}>\epsilon\right) \leq \delta ~~ \forall (X, X') \in {\cal E}$ is equivalent to \cref{eq:epsilon_delta_relation}.
It is important to remark that  prior statistical information about $X$ is generally available, i.e., $f(X)$ is a prior belief or distribution of $X$. The following remark explains how  \Cref{defn:e-d-privacy} is sufficient in  general.
\begin{rem}
If the analyst operates in the Bayesian setting and there is a statistical prior distribution on the possible outcomes for $X$, i.e., a probability model on $\mathcal{X}$, a more meaningful privacy loss definition than \cref{eq:privacycosnt} is expressed in terms of posterior distributions $f(X|\qtilde)$. Note, however, that:
\begin{align}
    \ln  \frac{f(X | \qtilde)}{f(X' |\qtilde)}=L_{\XX} ({\qtilde})+\ln\frac{f(X)}{f(X')}
\end{align}
where $f(X)$ is the prior distribution. This means that $(\epsilon,\delta)$ private according to \Cref{defn:e-d-privacy} then:
\begin{align}
\delta
 &\geq
\sup_{X\in{\cal X}}\sup_{X'\in{\cal X}^{(1)}_{X}} Pr\left(
   L_{\XX}(\qtilde)>\epsilon-\ln\frac{f(X)}{f(X')}
 \right).
\end{align}
\end{rem}

The implication is that the privacy loss distribution is inherently a function of the statistics of $L_{\XX}(\qtilde)$, even when there are priors, which makes the extension to the case where priors are given relatively straightforward. For this reason, we will continue the discussion considering the latent information, or label $X$, as deterministic and unknown. 

\subsection{ \response{Privacy under accuracy constraints}}
One of the aspects in which our framework differs from most other DP work is that we consider the constraint for the \textit{data owner} to guarantee a certain level of accuracy in the query response, we next define \textit{query accuracy}. This is different from the notion of \textit{query sensitivity}, which is an intrinsic property of the data,  and it is a measure of the utility of the DP query response.
 
For a {\it continuous query}, the {\it accuracy} is a measure of how dissimilar the answer to the query is, and it is desired to have $\qtilde\approx \mathbbm{Q}(\dbm)$. A possible simple measure, the average mean squared (MS) error per entry, is specified in the following:
\begin{defn}[Mean Square Error Accuracy]\label{def:contaccuracy}
For a continuous function $\mathbbm{Q}: \mathbb{R}^n \rightarrow \mathbb{R}^k$:
\begin{align}\label{eq:ms-sensitivity}
    \rho_{\mathbbm{Q}|X} =\frac 1 k \mathbb{E}[\|\qtilde-\mathbbm{Q}(\dbm)  \|_2^{2}],
\end{align}
is the \textbf{expected mean-square query accuracy}, 
where $\dbm \sim f(\dbm|X) $ and $X\in \mathcal{X}$ is the information to hold private.
\end{defn}

In our work, we adopt the \Cref{defn:e-d-privacy} for $(\epsilon,\delta)$ privacy along with accuracy constraint and thereby define the overall privacy in our framework as follows:
\begin{defn}\label{def:DPA}
A randomized algorithm $\mathbb{A}_{Q}$ with  outcome $\qtilde = \mathbbm{A}_{\mathrm Q}(\dbm | X)$, is $(\epsilon,\delta)$-\textbf{private} meeting an \textbf{accuracy budget} $\rho$ for a query $\mathbbm{Q}$ iff the condition in \cref{eq:epsilon_delta_relation} holds $\forall (X, X') \in {\cal E}$ along with $\rho_{\mathbbm{Q}|X} \leq \rho ~~\forall X \in {\cal X}$. 
\end{defn}

\subsection{\response{Additive noise mechanism}}
The most popular designs of DP algorithms amount to adding random noise to the true query, i.e.:
\begin{equation}
    \qtilde=\bm q+\bm{\eta},
\end{equation}
where $\bm{\eta}$ is drawn from a family of distributions that facilitate the calculation of the $(\epsilon,\delta)$ curves; zero mean Gaussian and Laplacian noise are the most frequent choices. Furthermore, the entries of the vector $\bm \eta$ are independent and identically distributed (i.i.d.) and they are independent \response{of the label} $X$, i.e. $f(\bm \eta|X)=f(\bm \eta)$. This makes only one parameter,  the variance, in the Gaussian and Laplacian distributions available as a degree of freedom, which is set once one defines the values of $(\epsilon,\delta)$. 
 
\response{In this case of additive noise mechanism,} the average MSE accuracy can be expressed as a function of the conditional mean and the covariance of the noise:
\begin{align}
    \rho^{\text{MSE}}_{\mathbbm Q|X}&=\frac 1 k \mathbbm{E}[
    \|\bm \eta\|^2]
    =\mathbb{E}_{\bm{q}|X}\left[\mathrm{Tr}(\bm \Sigma_{\bm{\eta}|\bm{q}})+
 \|\bm \mu_{\bm{\eta}|\bm{q}}\|^2\right]\\
 &= \frac 1 k \mathrm{Tr}(\bm \Sigma_{\bm \eta|X})+ \frac 1 k \|\bm \mu_{\bm{\eta}|X}\|^2,
 \label{eq:MS-accuracy-constr}
\end{align}

\subsection{\response{Case of deterministic queries conditioned on class-label}}
\par In contrast to our framework, the vast majority of the DP literature considers queries that are \response{\emph{deterministic}}  because they do not \response{explicitly} consider the distribution of data but only consider the sensitivity or the range. The true query answer is also deterministic, given $X$. \response{Therefore, as an illustration, we consider deterministic queries and calculate the $(\epsilon,\delta)$ curves when the data is deterministic given the label $X$. We will see that it corresponds to the classical DP additive Gaussian noise mechanism.}

 Let us denote deterministic query as 
$\bm q \equiv \mathbbm{Q}(X)$. The DP algorithm adds zero mean random noise that is independent of $X$. In this case, the noise is added to a constant, i.e.,
\begin{align}
    \qtilde=\mathbbm{Q}(X)+\bm{\eta}~~\Rightarrow~~f(\qtilde|X)=f_{\bm \eta}(\qtilde-\mathbbm{Q}(X)),
\end{align}
 which implies that for any pair $X,X'$  $f(\qtilde|X)$ and $f(\qtilde|X')$ differ only in their means, $\mathbb{Q}(X)$ and $\mathbb{Q}(X')$. 
 Let:
 \begin{align}
     \bm \mu_{\XX}\triangleq \mathbb{Q}(X)-\mathbb{Q}(X').\label{eq:mu_xxDQ}
 \end{align}
Now, query sensitivity for this case \response{can be defined as follows}:
\begin{defn}[Deterministic Query Sensitivity]\label{def:query-sensitivity}
The sensitivity of a deterministic query about the data $X$ is:
\begin{align}\label{eq:query-sensitivity}
    \Delta_p\triangleq \sup_{X\in {\cal X}}\sup_{X'\in {\cal X}^{(1)}_{X}}\|
\mathbbm{Q}(X)-\mathbbm{Q}(X')\|_p. 
\end{align}
where $\mathbbm{d}_Q \triangleq \| \mathbbm{Q}(X)-\mathbbm{Q}(X')\|_p$ is an appropriate notion of distance as $\ell_p$ norm that measures how much the queries applied differ when the label is $X$ or $X'$. 
\end{defn}
Now,  the computation of the $(\epsilon,\delta)$ curves when the data is deterministic given the label $X$ is shown.   
\begin{exmp}[\response{Additive Gaussian noise mechanism for deterministic query given $X$}]
The $(\epsilon,\delta)$ privacy curve is entirely defined by the noise distribution and its change due to a shift in the mean:
\begin{align}
    Pr(L_{\XX}(\qtilde)>\epsilon)&=Pr\left(
    \ln\frac{f_{\bm \eta}(\qtilde-\mathbb{Q}(X))}{f_{\bm \eta}(\qtilde-\mathbb{Q}(X'))}>\epsilon
    \right)\\
    &\equiv
    Pr\left(\ln\frac{f_{\bm \eta}(\bm \eta)}{f_{\bm \eta}(\bm \eta +\bm \mu_{\XX})}>\epsilon
    \right).\label{eq:noise-only-classic-DP}
\end{align}
From \cref{eq:noise-only-classic-DP} it is apparent that the interplay of the noise distribution and the possible values for the offset $\bm \mu_{\XX}$ are all that is needed to establish the $(\epsilon,\delta)$ privacy trade-off.   
Note that in the case of independent  noise added to different dimensions of the query:
\begin{align}
    L_{\XX'}(\qtilde)=\sum_{i=1}^k L_{\XX'}(\tilde{q}_i).
\end{align}
Indicating $[\bm \mu_{\XX}]_i=\mu_i$, the expression in \cref{eq:noise-only-classic-DP} is equivalent to:
\begin{align}
    Pr\left(
    \sum_{i=1}^k \ln \frac{f(\eta_i)}{f(\eta_i+\mu_i)} >\epsilon
    \right)
\end{align}

Next we derive $(\epsilon,\delta)$ bounds for Gaussian zero mean i.i.d. noise,
In this case $\bm \eta\sim {\cal N}(\bm 0,\sigma_{\eta}^2\bm I)$, it is easy to show that \cite{balle2018improving}:
\begin{align*}
 L_{\XX}(\qtilde)=\frac{\bm \mu_{\XX}\T(\qtilde\!-\! \mathbb{Q}(X))}{\sigma_{\eta}^2}+\frac{\|\bm \mu_{\XX}\|^2}{2\sigma_{\eta}^2}, 
\end{align*}
which implies that the likelihood $ L_{\XX}(\qtilde)$ is also Gaussian:
\begin{align}
     L_{\XX}(\qtilde)\sim {\cal N}\left(\frac{\|\bm \mu_{\XX}\|^2}{2\sigma_{\eta}^2},
 \frac{\|\bm \mu_{\XX}\|^2}{\sigma_{\eta}^2}\right).
\end{align}
denoting by 
$Q(v)=\frac 1 {\sqrt{2\pi}}\int_{v}e^{-\frac{u^2} 2}~\mathrm{d}u$, then:
\begin{align}\label{eq:Pr(L>e)-equal-cov}
Pr(L_{\XX}(\qtilde )>\epsilon)=
Q\left(\frac{\epsilon -\frac{\|\bm \mu_{\XX}\|_2^2}{ 2\sigma_{\eta}^2}}{\frac{\|\bm \mu_{\XX}\|_2}{\sigma_{\eta}}}\right)
\end{align}
In this case, since the trend of the probability is a monotonic function of $\|\bm \mu_{\XX}\|_2$ (i.e., the Euclidean distance of the queries) a meaningful definition for query sensitivity in \Cref{def:query-sensitivity} is $\Delta_2=\sup_{X}\sup_{X'\in {\cal X}_{\X}^{(1)}}\|\bm \mu_{\XX}\|_2$, which implies:
\begin{align}
    \delta =
    Q\left(\frac{\epsilon -\frac{\Delta_2^2}{ 2\sigma_{\eta}^2}}{\frac{\Delta_2}{\sigma_{\eta}}}\right), 
\end{align}
and if we set a limit $\rho$ for the MSE, then $\sigma_{\eta}^2\leq \rho/k$. \\
\end{exmp}
\response{The discussion for deterministic queries naturally leads us to discuss the more general case of stochastic query responses for which we propose the publication mechanism. }
\response{Having given all the necessary definitions, we now introduce our method of publishing the query response. We specifically consider query responses that are continuous-valued and Gaussian distributed given  label $X$. This is because such models are ubiquitously used. }

\section{\response{Proposed Additive Gaussian Noise Mechanism for Stochastic Queries}}\label{sec:publishing_mechanism}
\response{The idea behind the proposed query publication methodology $\mathbb{A}_{Q}$ is simple: we also add Gaussian noise to the actual query response, but the noise $\bm{\eta}$ vector of the mechanism we propose $f_{\bm \eta|X}(\bm \eta|X)$, is  label $X$  dependent in general, and its entries of $\bm \eta$ are not i.i.d.    
This leaves us with several additional degrees of freedom to meet an accuracy constraint, which is a function of the noise statistics $f_{\bm \eta|X}(\bm \eta)$, as we will discuss in this section.}

 \par \response{When the query, $\bm{q}=\mathbbm{Q}(\dbm)$ is modeled as being an outcome of a random ensemble whose distribution depends on the  label $X$, sampling from the query distribution, $f(\bm q|X)$ can be considered a randomized mechanism for privacy. This is because the process of sampling from the query distribution $f(\bm q|X)$ is inherently random and provides a certain level of privacy wherein the $(\epsilon, \delta)$ curves can be calculated without adding noise at all, i.e. using $f(\bm q|X)$ in lieu of $f(\qtilde|X)$. 
 If the query responses without noise do not reveal what the underlying $X$ is then there is no need to alter the data, i.e., the query response is naturally private without employing any random mechanism. }
 \par \response{This phenomenon has been studied  as `privacy for free' in \cite{wang2015privacy} where the sensitive dataset is used to learn model parameters in a Bayesian manner. The authors show that releasing one sample from the posterior parameter distribution is differentially private. The query, which is the parameter to be learned, is stochastic which is indeed the same situation that we discuss. } 
 \response{ The difference, however, arises from the proposed definition of class-label privacy. One cannot control the generative mechanism for the data. In the worst case, query distributions arising from neighboring labels  can be quite distinguishable which means that we would need to add suitable noise to further reduce $\delta$. }
The addition of suitable noise can mask the hypotheses and yield a lower $\delta$ for a certain $\epsilon$, sacrificing the accuracy of the response. 
Now, the distribution of the query response after the addition of noise is computed as follows:
\begin{align}
    f(\qtilde|\bm{q},X)=
    f(\qtilde|X)&=\int f_{\bm \eta|X}(\qtilde-\bm{q}|X)f(\bm q|X)~\mathrm{d}\bm q.
\end{align}
The analytical calculation of this convolution integral would be non-trivial in general. 

This is why, in this paper, we discuss the case where the query response given the label $X$ is Gaussian and the noise  is also Gaussian. We do this  because such models are ubiquitously used.
In mathematical terms, if $\bm q\sim {\cal N}(\bm \mu_{\X}, \bm \Sigma_{\X})$ (where $\bm \mu_{\X}$ is the mean and $\bm \Sigma_{\X}$ is the covariance matrix of $\bm q$ with label $X$), and 
$\bm \eta\sim {\cal N}(\bm \mu_{\eta}, \bm \Sigma_{\bm \eta})$, then $\qtilde \sim {\cal N}(
\bm \mu_{\X}+\bm \mu_{\bm \eta},\bm \Sigma_{\X}+\bm \Sigma_{\bm \eta}
)$.
 To streamline the notation we will use the conventions:
\begin{align}
\bm \mu_{\X}&\triangleq\bm \mu_{\bm q|X},~~\tilde{\bm \mu}_{\X}\triangleq\bm \mu_{\qtilde|X} = \bm \mu_{\X}+\bm \mu_{\bm \eta|X},\label{eq:tilde-mu}
\\
\bm \Sigma_{\X}&\triangleq\bm \Sigma_{\bm q|X},~ 
\tilde{\bm \Sigma}_{\X}\triangleq\bm \Sigma_{\qtilde|X} = \bm \Sigma_{\X}+\bm \Sigma_{\bm \eta|X}.\label{eq:tilde-Sigma}
\end{align}
In this case, we bound explicitly the probability of the privacy loss random variable for  both for the query $\bm q$ itself,  $Pr(L_{\XX}(\bm q)>\epsilon)$, and for the published query $\qtilde$, i.e. $Pr(L_{\XX}(\qtilde)>\epsilon)$. 
For the analysis, we  use the following fact:

\begin{prop}\label{prop:log-likelihood-gauss} Let the  log-likelihood be $L_{\XX}(\bm q) = \ln \frac{f(\bm q|X)}{f(\bm q|X')}$. For the queries that  are drawn when the private information is $X$, i.e. $\bm q\sim {\cal N}(\bm \mu_{\X},\bm \Sigma_{\X})$, equivalently we have:
\begin{align}
    Pr(L_{\XX}(\bm q)>\epsilon)\equiv Pr(L_{\XX}(\bm \xi)>\epsilon)
\end{align}
where $\bm \xi$  and $L_{\XX}(\bm \xi)$ are:
\begin{align}
    \bm \xi&\triangleq\bm U_{\XX}\T\bm \Sigma_{X}^{-1/2}(\bm q-\bm \mu_{\X})~\sim~{\cal N}(\bm 0,\bm I);\label{eq:xi}\\
    L(\bm \xi)&\triangleq-\frac 1 2 \ln|\bm \Gamma_{\XX}|+\frac{1}{2}\bm \xi\T(\bm \Gamma_{\XX}\!-\!\bm I)\bm\xi\nonumber\\
    &-\bm \mu_{\XX}\T\bm \Gamma_{\XX}\bm \xi+\!\frac{1}{2}\bm \mu_{\XX}\T\bm \Gamma_{\XX}\bm \mu_{\XX},
    \label{eq:L(xi)}\\
\text{and} ~\bm \mu_{\XX} &\triangleq \bm U_{\XX}\T \bm \Sigma_{\X}^{-1/2}(\bm \mu_{\X'}-\bm \mu_{\X}) \label{eq:mu_XX'}. 
\end{align}
$\bm U_{\XX}, \bm \Gamma_{\XX}$ are the eigenvectors and eigenvalue matrices respectively,
\begin{align}
\bm U_{\XX}\bm \Gamma_{\XX}\bm U_{\XX}\T\triangleq\bm \Sigma_{\X}^{1/2}\bm \Sigma_{\X'}^{-1}\bm \Sigma_{\X}^{1/2}, 
\end{align}
with the diagonal matrix $\bm \Gamma_{\XX}=\mathrm{diag}(\bm \gamma_{\XX})$ and the vector  $\bm \gamma_{\XX}$, whose entries are the eigenvalues in descending order and the unitary matrix of eigenvectors, $\bm U_{\XX}$.
\end{prop}
\begin{proof} See \Cref{proof:L(xi)}.
\end{proof}
Naturally, the same expressions hold for $\qtilde$ except that to define $\tilde{\bm U}_{\XX}, \tilde{\bm \Gamma}_{\XX}$ and $\tilde{\bm \mu}_{\XX}$ we are using \cref{eq:tilde-mu} and \cref{eq:tilde-Sigma} and corresponding expressions for $X'$. For the conditionally Gaussian case, in order to create a better level of $(\epsilon, \delta)$  privacy the noise design should reduce the difference of the means and covariances by adding noise. 
 In particular, the quadratic term in \cref{eq:L(xi)} becomes zero if $\bm \Gamma_{\XX}=\bm I$, which is the case when the noise covariances of the two hypotheses under $X$ and $X'$ are the same, and the remaining linear and constant terms go to zero if the difference between the mean vectors is zero. 
In the case of $\bm \Gamma_{\XX}=\bm I$, $Pr(L_{\XX}(\bm q)>\epsilon)$ can be computed in closed form as shown
\begin{corollary}\label{cor:equal_cova}
Assume $\bm \Sigma_{\X}\!=\!\bm \Sigma_{\X'}$, so that
$\bm \Gamma_{\XX}\!=\!\bm I$. Then:
\begin{align}
    L_{\XX'}(\bm \xi)\sim {\cal N}\left(\frac 1 2\|\bm \mu_{\XX}\|^2, \|\bm \mu_{\XX}\|^2\right).
\end{align}
and with $Q(v)=\frac 1 {\sqrt{2\pi}}\int_{v}e^{-\frac{u^2} 2}~\mathrm{d}u$, then:
\begin{align}\label{eq:Pr(L>e)-equal-cov-2}
Pr(L_{\XX}(\bm q)>\epsilon)=
Q\left(\frac{\epsilon -\frac{\|\bm \mu_{\XX}\|^2} 2}{\|\bm \mu_{\XX}\|}\right)
\end{align}
which is a monotonically increasing function of $\|\bm \mu_{\XX}\|^2$.
\end{corollary} 

However, more generally when $\bm \Sigma_{\X}\!\neq \!\bm \Sigma_{\X'}$, we use the Chernoff bound for $Pr(L_{\XX}(\bm q)>\epsilon)$ to help evaluate the $(\epsilon,\delta)$ privacy levels. The same expression can also be used as a bound for $Pr(L_{\XX}(\qtilde)>\epsilon)$ using $\tilde{\bm \mu}_{\XX}$, and $\tilde{\bm U}_{\XX}, \tilde{\bm \Gamma}_{\XX}$ instead: 
\begin{lem}\label{lem:chernoff-bound}
For all $s>\max(1,\gamma_1)$ where $\gamma_1=\lambda_{\max}(\bm\Sigma_{\X}^{1/2}\bm \Sigma_{\X'}^{-1}\bm\Sigma_{\X}^{1/2})$, it holds:
\begin{align}\label{eq:chernoff}
    Pr(L_{\XX}(\bm q)>\epsilon)&\leq 
    \frac{(s\!-\!1)^{\frac{k}{2}}}
{|\bm \Gamma_{\XX}|^{\frac{1}{2(s\!-\!1)}}|s\bm I-\bm \Gamma_{\XX}|^{1/2}}\\
&\times e^{-\frac{\epsilon}{(s\!-\!1)}+\frac{s}{2(s\!-\!1)}{\bm \mu}_{\XX}\T(s\bm I-\bm \Gamma_{\XX})^{-1}\bm \Gamma_{\XX}{\bm \mu}_{\XX}}\nonumber
\end{align}
\end{lem}
\begin{proof}
See \Cref{app:chernoff-bound}.
\end{proof}
In order to understand the trends in the bound, we can specify it for some value of $s$ in a corollary of~\Cref{lem:chernoff-bound}:
\begin{corollary}\label{cor:chernoff-bound}
For $s\gg \max(1,\gamma_1)$, 
\begin{align}
    Pr(L_{\XX}(\bm q)>\epsilon)&\lessapprox  \frac{1}{|\bm \Gamma_{\XX}|^{\frac{1}{2s}}} e^{-\frac{\epsilon}{s}} e^{\frac{\bm \mu_{\XX}\T\bm \Gamma_{\XX'}\bm \mu_{\XX}} {2s}} \\
    & = e^{\frac{1}{2s} \left[\bm \mu_{\XX}\T\bm \Gamma_{\XX'}\bm \mu_{\XX} - \ln|\bm \Gamma_{\XX}| \right] }e^{-\frac{\epsilon}{s}} 
\end{align}
\end{corollary}
The bounds suggest that, more generally, $Pr(L_{\XX}(\bm q)>\epsilon)$ increases monotonically, as $\bm \mu_{\XX}\T\bm \Gamma_{\XX'}\bm \mu_{\XX}-\ln|\bm \Gamma_{\XX}|$ increases, which is consistent with the case where the expression is exact and $\bm \Gamma_{\XX'}=\bm I$. This metric will be leveraged in our optimal design. In particular, we can assert that for the conditionally Gaussian case we can use the following $(\epsilon, \delta)$ bound:
\begin{corollary}\label{cor:Gaussian_sensitivity}
Let the Gaussian sensitivity be defined as:
\begin{align}
    \Delta_{G}=\sup_{\X}\sup_{\X'\in {\cal X}_{\X}^{(1)}}
    \frac{1}{2}\tilde{\bm \mu}_{\XX}\T\tilde{\bm \Gamma}_{\XX'}\tilde{\bm \mu}_{\XX}-\frac{1}{2}\ln|\tilde{\bm \Gamma}_{\XX}|
\end{align}
The following trend is a $(\epsilon, \delta)$ bound for the privacy loss:
\begin{align}
    \delta_{\qtilde}^{\epsilon}\leq \delta=
   e^{\frac{\Delta_{G}-\epsilon}{s}}.  
\end{align}
\end{corollary}
The bound is simple but quite loose; however, it does help to identify the worst-case scenario which, in turn, helps to optimize the noise parameters. 

\subsection{\response{Optimal design of additive noise parameters}}
In this section, we propose an algorithm to optimize the parameters of the additive noise mechanism algorithm $\mathbbm{A}_{\mathbbm{Q}}(\bm d)$.  
Considering the expression of mean and covariance of $\qtilde$ in \cref{eq:tilde-mu,eq:tilde-Sigma} the expression we derived for $\bm q$
can be applied to derive $Pr(L_{\XX}(\qtilde)>\epsilon)$ by defining the corresponding vector:
\begin{align}
    \tilde{\bm \mu}_{\XX}\triangleq
     \tilde{\bm U}_{\XX}\tilde{\bm \Sigma}_{\X}^{-\frac 1 2}(\tilde{\bm \mu}_{\X'}\!-\!\tilde{\bm \mu}_{\X}), \label{eq:define_mu_XX_prime}
\end{align} 
and define the unitary matrix $\tilde{\bm U}_{\XX}$ and diagonal matrix $\tilde{\bm \Gamma}_{\XX}$ through the following eigenvalue decomposition:
\begin{align}
    \tilde{\bm U}_{\XX}\tilde{\bm \Gamma}_{\XX}\tilde{\bm U}_{\XX}\T\triangleq 
    \tilde{\bm \Sigma}_{\X}^{\frac 1 2}\tilde{\bm \Sigma}_{\X'}^{-1}\tilde{\bm \Sigma}_{\X}^{\frac 1 2}. \label{eq:Gamma_XX_prime}
\end{align}
From the analysis in the previous section, and \Cref{cor:chernoff-bound,cor:Gaussian_sensitivity}, it seems that a good surrogate metric for achieving the best $(\epsilon,\delta)$ privacy is:
\begin{align}
\!\!\!\!\tilde{\bm \mu}_{\XX}\T\tilde{\bm \Gamma}_{\XX'}\tilde{\bm \mu}_{\XX} - \ln|\tilde{\bm \Gamma}_{\XX}| = (\tilde{\bm \mu}_{\X'}- \tilde{\bm \mu}_{\X})\T \tilde{\bm \Sigma}_{\X'}^{-1} \!(\tilde{\bm \mu}_{\X'} \!-\! \tilde{\bm \mu}_{\X})\! -\! \ln|\tilde{\bm \Gamma}_{\XX}|
\end{align}
In fact, with $s\gg \max (\gamma_1,1)$  for all $\tilde{\bm \Gamma}_{\XX}$:
\begin{align}
    \log \delta_{\qtilde}^{\epsilon}
    \lessapprox& \sup_{\X\in {\cal X}}\sup_{\X'\in {\cal X}_{\X}^{(1)}} \!\left(
    -\frac{\epsilon}{2s}\!+\!\frac{\tilde{\bm \mu}_{\XX}\T\tilde{\bm \Gamma}_{\XX}\tilde{\bm \mu}_{\XX}\!-\!\ln|\tilde{\bm \Gamma}_{\XX}|}{s}
    \right)\nonumber\\
    =&\!-\!\frac{\epsilon}{2s}
    \!+\!\frac 1 s \sup_{\X\in {\cal X}}\sup_{\X\in {\cal X}_{\X}^{(1)}}(\tilde{\bm \mu}_{\XX}\T\tilde{\bm \Gamma}_{\XX}\tilde{\bm \mu}_{\XX}\!-\!\ln|\tilde{\bm \Gamma}_{\XX}|).
\end{align}
We therefore seek to find the noise means $\bm \mu_{\bm \eta|\X}$ and covariances $\bm \Sigma_{\bm \eta|\X}$ that solve the following problem:
\begin{align}
    \min_{\bm \mu_{\eta|\X},\bm \Sigma_{\bm \eta|\X},\forall X}& \left(\sup_{\X\in {\cal X}}\sup_{\X'\in {\cal X}_X^{(1)}}\tilde{\bm \mu}_{\XX}\T\tilde{\bm \Gamma}_{\XX}\tilde{\bm \mu}_{\XX}\!-\!\ln|\tilde{\bm \Gamma}_{\XX}|\right),\nonumber\\
    \mbox{subj. to}&~ \tilde{ \bm \mu}_{\X} = \bm \mu_{\X}+\bm \mu_{\bm \eta|X}, 
    \tilde{\bm \Sigma}_{\X} = \bm \Sigma_{\X}+\bm \Sigma_{\bm \eta|X}, \\
    \rho & = \frac 1 k \mathrm{Tr}(\bm \Sigma_{\bm \eta|X})+ \frac 1 k \|\mu_{\bm{\eta}|X}\|^2,
    \label{eq:opt-noise}
\end{align}
For the time being, we assume that noise means $\bm \mu_{\bm \eta|\X} = \bm{0}$.
Expanding the terms within the $\sup$, we get,
\begin{align}
\tilde{\bm \mu}_{\XX}\T\tilde{\bm \Gamma}_{\XX}\tilde{\bm \mu}_{\XX}\!-\!\ln|\tilde{\bm \Gamma}_{\XX}|&=  (\tilde{\bm \mu}_{\X'} \!-\!\tilde{\bm \mu}_{\X})\T \bm \tilde{\bm \Sigma}_{\X'}^{-1} (\tilde{\bm \mu}_{\X'} \! - \!\tilde{\bm \mu}_{\X})\!-\!\ln \frac{| \tilde{\bm \Sigma}_{\X}  |}{| \tilde{\bm \Sigma}_{\X'}  |}\nonumber\\
& =  (\tilde{\bm \mu}_{\X'}\!-\!\tilde{\bm \mu}_{\X})\T ( \bm \Sigma_{\X'}\!+\!\bm \Sigma_{\bm \eta|\X'}   )^{-1} (\tilde{\bm \mu}_{\X'}\!-\! \tilde{\bm \mu}_{\X}) \nonumber \\
&\quad - \ln \frac{|\bm \Sigma_{\X}+\bm \Sigma_{\bm \eta|\X}  |}{|\bm \Sigma_{\X'}+\bm \Sigma_{\bm \eta|\X'}  |} \label{eq:surrogate_metric}
\end{align}
The surrogate metric in \cref{eq:surrogate_metric} is convex in $\tilde{\bm \Sigma}_{\X}$  when all the other terms are kept constant. It is also  convex in the variable $\tilde{\bm \Sigma}_{\X'}^{-1}$. 
We continue the discussion by setting the optimization variable as $\bm{A}_{\X} \triangleq \tilde{\bm \Sigma}_{\X}^{-1}$ and later estimate $\bm \Sigma_{\bm \eta|\X}$ by subtracting  ${\bm \Sigma}_{\X}$ from $\tilde{\bm \Sigma}_{\X}$  and projecting onto the set of positive semi-definite (PSD) matrices. 

Now, let us focus on the overall cost of the optimization problem in \cref{eq:opt-noise}. 
 The  terms are $\forall X \in \mathcal{X}, ~ X' \in {\cal X}_X^{(1)}$:
\begin{align}
& g_{\XX} 
\triangleq (\tilde{\bm \mu}_{\X'} - \tilde{\bm \mu}_{\X})\T \bm{A}_{\X'} (\tilde{\bm \mu}_{\X'} - \tilde{\bm \mu}_{\X}) - \ln
\frac{| \bm{A}_{\X'}  |}{| \bm{A}_{\X}  |}.\label{eq: g_XX}
\end{align}
The cost associated with the minimization then becomes 
\begin{align}
J = \sup_{\X   } \sup_{  \X' \in  {\cal X}_X^{(1)}   }    g_{\XX} 
\end{align}
{Note that there is no symmetry in the surrogate metric, $ g_{\XX}   
\neq  g_{\X'\!\X} $ and therefore, we need to consider them separately.}
One can think of a block-coordinate descent algorithm in order to reach one of the minima (local or global).  The algorithm we propose is given in \Cref{alg: A_update}. 
\begin{algorithm}
    \KwResult{$\bm{A}^{\star}_{\X}, ~~ \forall X \in {\cal X}$}
    Initialize $\tilde{\bm \Sigma}_{\X}^{(0)}  = \bm{\Sigma}_{\X} + (\rho/k) \bm{I}$, $\alpha_{X}^{(0)}=0.1 ~\forall X $ and $\bm{A}_{\X}^{(0)} = \left(\tilde{\bm \Sigma}_{\X}^{(0)}\right)^{-1}$ for all $X \in {\cal X}, t=0$.
    
    \While{ $\alpha_{\X_{\bullet}}^{(t)} >0$}
    {
    Identify  pair of labels: 
    $$(X_{\bullet}  ,X'_{\bullet}) \!=\! \arg_{ (X, X') } \sup_{\X   } \!\sup_{  \X' \in  {\cal X}_X^{(1)}   } \!  \{ g_{\X \X'}(t)\}.$$
    
    Compute $J^{(t)} = \! g_{\X_{\bullet}\X_{\bullet}'}(t)$ using  $(\bm{A}_{\X_{\bullet}}^{(t)}, \! \bm{A}_{\X_{\bullet}}^{(t)})$ in \cref{eq: g_XX}.
    
	Update $\bm{A}_{\X_{\bullet}}^{(t+1)} $ with $\alpha_{\X_{\bullet}}^{(t)}$ in \Cref{lem:alpha_t}:
	\begin{align}
    	&\hspace{-3pt}\! \bm{A}_{\X_{\bullet} }^{(t+1)}\! \leftarrow\! \bm{A}_{\X_{\bullet}}^{(t)} - \alpha_{\X_{\bullet}}^{(t)} \left[\nabla_{ \bm{A}_{\X_{\bullet} }   }  g_{\X_{\bullet} \X_{\bullet}}(t) \right], \label{eq:A_update}
	\end{align}
	where
	$$
        \nabla_{ \bm{A}_{\X_{\bullet} }}   g_{\X_{\bullet}\X_{\bullet}}(t) \!=\!\! (\tilde{\bm \mu}_{ \X_{\bullet}}^{(t)} \!-\! \tilde{\bm \mu}_{\X_{\bullet}}^{(t)}) (\tilde{\bm \mu}_{ \X_{\bullet}}^{(t)} \!-\! \tilde{\bm \mu}_{\X_{\bullet}}^{(t)})\T - (\bm{A}_{\X_{\bullet}}^{(t)})^{-1}.
        $$
	
	Project $\bm{A}_{\X_{\bullet}(t) }(t+1)$ onto the set of positive-semi definite matrices:
	\begin{align}
	     \bm{A}_{\X_{\bullet}}^{(t+1)} \leftarrow \text{Proj}_{{\cal S}_k^{+}} \{ \bm{A}_{\X_{\bullet} }^{(t+1)} \}.
	\end{align}
	 $$ t = t+1$$
    }

    \caption{Estimation of $\bm{A}_{\X} = \tilde{\bm \Sigma}_{\X}^{-1}, \forall X \in {\cal X}$ } \label{alg: A_update}
\end{algorithm}
Having obtained an $\bm{A}_{\X}$, one can determine the noise covariance matrix $\bm \Sigma_{\bm \eta|\X}$ such that the constraint $\mathrm{Tr}(\bm \Sigma_{\bm \eta|X})$ is satisfied as follows: 
\begin{align}
 \bm \Sigma_{\bm \eta|X}  =  \frac{\rho ~ \text{Proj}_{{\cal S}_k^{+}} \{ \left(\bm{A}_{\X}^{\star}\right)^{-1} - \bm{\Sigma}_{\X} \}}{\mathrm{Tr} \left(  \text{Proj}_{{\cal S}_k^{+}} \{ \left(\bm{A}_{\X}^{\star}\right)^{-1} - \bm{\Sigma}_{\X} \}   \right)} \label{dpnoisecov}
\end{align}
where $\text{Proj}_{{\cal S}_k^{+}}(\bm X)$ is the projection of the matrix $\bm X$ onto the set of positive semi-definite matrices.
We know that the optimization problem of interest is a non-convex one and what we seek is an improvement over the standard Gaussian mechanism, which amounts to adding white noise to the query response. Our algorithm initializes the noise covariance matrices to be white. Then, the minimization we perform with respect to different variables is such that there is always a \emph{decrease} in the cost function. For this, it is pertinent that we choose the appropriate step-size $\alpha_{\X_{\bullet}(t) }$ in order to converge to a local minimum. 
\begin{lem}
  \label{lem:alpha_t}
  Let  $J^{(t)} = g_{\X_{\bullet} \X_{\bullet}}(t) $ and the variable to be updated $\bm{A}_{\X_{\bullet} }^{ (t)} $. 	If the step size $\alpha_{\X_{\bullet}^{(t)} }$ at iteration $t$ satisfies:
  \begin{align}\label{eq:alpha}
  \alpha_{\X_{\bullet}}^{(t)} &\leq ~~\max \left\{ 0,\min_{X \in {\cal X}_{ \X_{\bullet}(t)  }^{(1)}    } \{   b_{\X_{\bullet}\X }(t),   d_{\X_{\bullet}\X }(t)   \} \right\} ~~~~
  \end{align}
   then, the update to $\bm{A}_{\X_{\bullet} }^{ (t)} $ given in \cref{eq:A_update} will lead to $J^{(t)} - J^{(t+1)} \geq 0 $.
The constants in \cref{eq:alpha} are,
  \begin{align}
 \hspace{-2pt} b_{\X_{\bullet}\X }(t) &= \frac{  \tilde{\bm{v}}_{\X\X_{\bullet} }\T  \bm{A}_{\X_{\bullet} }^{ (t)} \tilde{\bm{v}}_{\X\X_{\bullet} }\! - \!  \tilde{\bm{v}}_{\X\X_{\bullet} }\T \bm{A}_{\X }^{(t)} \tilde{\bm{v}}_{\X \X_{\bullet} } }{  \mathrm{Tr} \left( (\bm{A}_{\X_{\bullet}}^{(t)})^{-2}  \right) -  \tilde{\bm{v}}_{\X\X_{\bullet} }^{\T}  (\bm{A}_{\X_{\bullet} }^{(t)})^{-1}  \tilde{\bm{v}}_{\X\X_{\bullet} }     } \label{eq:b_xx}\\
  & + \frac{\! \ln |\bm{A}_{\X_{\bullet}}^{(t)}   | +   \ln |\bm{A}_{\X }^{(t)}   |\!-\!2  \ln  |\bm{A}_{\X_{\bullet} }^{ (t)}  |   } {\mathrm{Tr} \left( (\bm{A}_{\X_{\bullet}}^{(t)})^{-2} \right) -  \tilde{\bm{v}}_{\X\X_{\bullet} }^{\T}  (\bm{A}_{\X_{\bullet}}^{(t)})^{-1}  \tilde{\bm{v}}_{\X\X_{\bullet} }       },  \nonumber\\
 \hspace{-2pt}  d_{\X_{\bullet}\X }(t)  &= \frac{ \tilde{\bm{v}}_{\X\X_{\bullet} }\T  \bm{A}_{\X_{\bullet} }^{(t)} \tilde{\bm{v}}_{\X\X_{\bullet} } \!  + \! \ln |\bm{A}_{\X_{\bullet}}^{(t)}  | \!-\!   \ln |\bm{A}_{\X }^{(t)} |        }{\tilde{\bm{v}}_{\X\X_{\bullet} }\T (\bm{A}_{\X_{\bullet}}^{(t)})^{-1} \tilde{\bm{v}}_{\X\X_{\bullet} }\! -\! \mathrm{Tr} \left( (\bm{A}_{\X_{\bullet}}^{(t)})^{-2}   \right) \!-\!  ( \tilde{\bm{v}}_{\X\X_{\bullet} }\T \tilde{\bm{v}}_{\X\X_{\bullet} } )^2     }, 
  \end{align}
 where $\tilde{\bm{v}}_{\XX } = (\tilde{\bm \mu}_{ \X' } \!- \!\tilde{\bm \mu}_{\X})$.
\end{lem}
\begin{proof}
Proof can be found in \Cref{app:step_size_t}.
\end{proof}

\section{Application to ARMA forecasts}\label{sec:ARMA_forecasts}
In this section, we first discuss how one could apply the proposed DP method for class-based data to stationary Gaussian processes, where each class-label $X$ is associated with a Gaussian process with a certain mean and covariance. Then we discuss DP-based ARMA forecasts, where each class has a different set of ARMA model parameters. Finally, we outline how one could produce DP power consumption forecasts while hiding the identity of a certain household. \response{This is necessary since revealing household power consumption patterns can lead to inference about personal activities, thereby violating privacy.}
\par Let $x[k]$ be a stochastic process, whose statistics depend on a label $X$ we want to conceal;  let us define the vector $\bm x$ of samples that contains both $K$ observed samples $\bm x^o$ and $T$ future samples $\bm x^{f}$ whose forecast we want to share without revealing the label $X$, i.e.:
\begin{equation}
 \bm  x=[x[0], \ldots, x[K+T-1]]\T
 =
 \begin{bmatrix}
 \bm x^o\\
 \bm x^f
 \end{bmatrix}.
\end{equation}
\response{The query here is the forecast, $\bm x^f$.}
A well known result from statistics is that the minimum-mean squared error (MMSE) estimator of the forecast is the conditional expectation of $\bm x^f$, given $\bm x^o$, i.e.:
\begin{equation}
 \bm q\equiv
 \hat{\bm x}^f=\mathbb{E}[\bm x^f|\bm x^o].
\end{equation}
For a Gaussian process the statistics of 
$\bm x^f$ conditioned on $\bm x^o$ are also Gaussian. To compute the mean and covariance of the conditional mean we need to also define the covariance of $\bm x$:
\begin{equation}
\mathbb{E}[\bm x]=\begin{bmatrix}
\bm \mu_{\X}^o\\
\bm \mu_{\X}^f
\end{bmatrix},~~~~
\mathbb{E}[\bm x\bm x^\top]
=
\begin{bmatrix}
\bm \Sigma_{\X}^{o} & \bm \Sigma_{\X}^{fo}\\\bm 
\Sigma_{\X}^{of}
&
\bm \Sigma_{\X}^{f}
\end{bmatrix}.
\end{equation}
This implies that for the query $\bm{q}$ under the class-label $X$, the mean and covariance are :
\begin{equation}
\begin{aligned}
\label{eq:Gauss-forecast}
    \bm \mu_{\X}&=
    \bm \mu_{\X}^f+
    \bm \Sigma_{\X}^{fo}
    (\bm \Sigma_{\X}^{o})^{-1}(\bm x^o-\bm \mu_{\X}^o)\\
    \bm \Sigma_{\X}&=
    \bm \Sigma_{\X}^f-
    \bm \Sigma_{\X}^{fo}
    (\bm \Sigma_{\X}^{o})^{-1}\bm \Sigma_{\X}^{of}
\end{aligned}
\end{equation}
At this point, it should be clear that it is possible to adopt the scheme we proposed to share a DP forecast for a Gaussian process. Note that the same formulas provide the optimum linear-least square forecast. 
The typical situation is that the covariance and mean are not known and need to be estimated from the data. This is why it is popular to assume a parametric model for the second and first order statistics. The ARMA model for the second order statistics is an example. 
\begin{figure}
    \centering
    \includegraphics[width=0.35\textwidth]{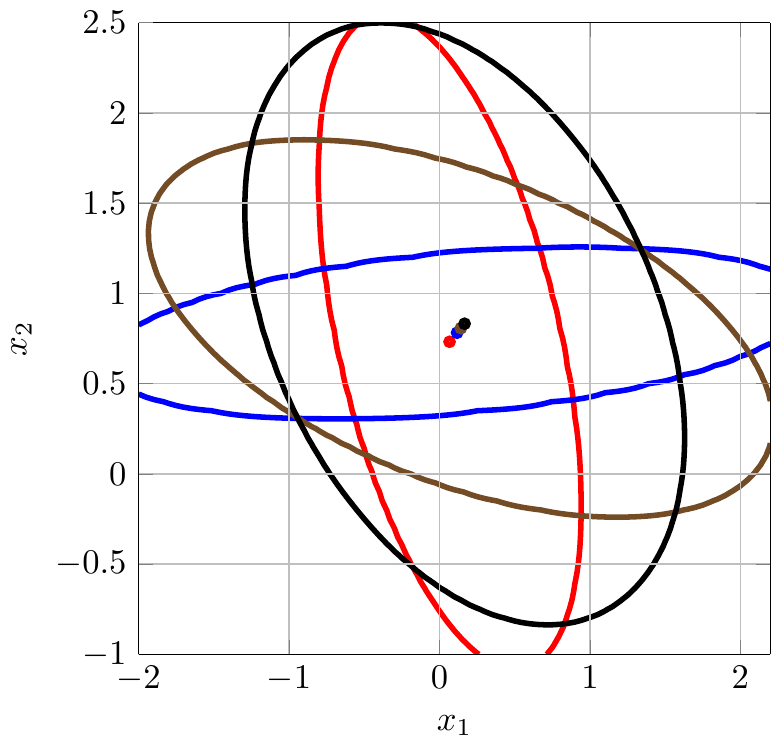}
    \caption{The contour plot of synthetic data.}
    \label{fig:contour_plot}
\end{figure}
\begin{figure*}[!htbp]
 \centering
 \includegraphics[width=\textwidth]{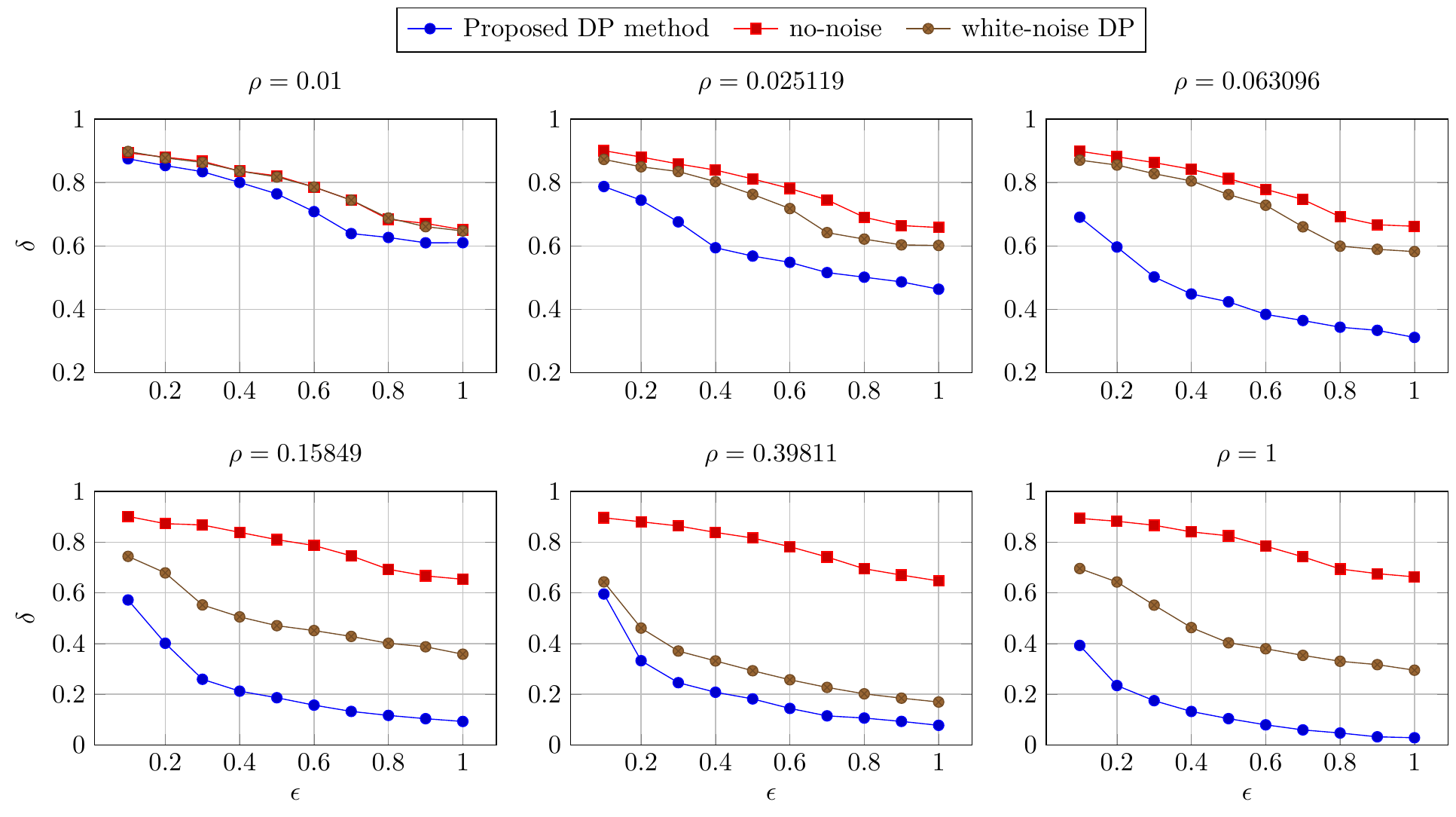}
\caption{The $(\epsilon,\delta)$-private curves with different \textbf{accuracy budgets} (see Def. \ref{def:DPA})  $\rho$ for the synthetic data.}  
  \label{synthetic_arma}
\end{figure*}

\subsection{ARMA forecasts}
An ARMA filter is a discrete time filter  whose frequency response can be parametrized as follows:
\begin{equation}
 H(e^{-j\omega})=
 \frac{\sum_{k=1}^m a_k e^{-j\omega k} }{\sum_{k=0}^n b_k e^{-j\omega k} }. \label{eq:ARMA_freq_resp}
\end{equation}
An ARMA process is a zero mean stationary process that is generated by filtering with an ARMA filter i.i.d. zero mean noise $\xi[k]$ with standard deviation equal to $1$; equivalently,
\begin{equation}
x[k]=-\sum_{i=1}^m a_i x[k-i]+\sum_{i=0}^n b_i \xi[k-i].   \label{arma}
\end{equation}
If the process is not zero mean, then $x[k]$ models the residual after recentering the actual process by subtracting the mean. From the estimate of the parameters of the ARMA filter one can calculate the impulse response $h[k]$ whose frequency response is given in \cref{eq:ARMA_freq_resp} (we skip the expression for brevity). With that, one can construct the Toeplitz matrix of the covariance matrix whose $ik^{\text{th}}$ entry is
\begin{equation}
\big[\mathbb{E}[\bm x\bm x^\top]\big]_{ik}=\sum_{i=-\infty}^{+\infty} h[i]h[k-i]
\end{equation}
and proceed to calculate the forecast. Note that it is possible to answer the query or compute $\bm{q}$ using the filter defined above. However, since our method of numerical calculation relies on the covariance of the process, we rely on the expression in  \cref{eq:Gauss-forecast} to compute the optimum noise distribution for the forecast $\bm{q}$. 
We summarize the process of adding optimal DP noise to ARMA forecasts in \Cref{alg: ARMADP}. 
\begin{algorithm}
\caption{Addition of DP noise to ARMA forecasts} \label{alg: ARMADP}
\SetKwInOut{Input}{input}
\SetKwInOut{Output}{output}
Compute the ARMA forecasting covariance matrix $\bm \Sigma_{\X}$ in \cref{eq:Gauss-forecast} $\forall X \in \mathcal{X}$.

Estimate $\bm{A}_{\X} \forall X \in \mathcal{X}$ according to  \Cref{alg: A_update}.

Obtain optimal DP noise covariance matrix $\bm \Sigma_{\bm \eta|\X}$ in \cref{dpnoisecov}.

Add the optimal DP noise  \response{to the forecast $\bm{q} = \hat{\bm x}^f=\mathbb{E}[\bm x^f|\bm x^o].$}, i.e., $\tilde{\bm{q}} = \bm{q} + \bm{\eta}$ where $\bm{\eta} \sim \mathcal{N}(\bm{0}, \bm \Sigma_{\bm \eta|\X} )$.
\end{algorithm}

\subsection{Application to power measurements data}\label{subsec:arma_power}
To validate our method in the numerical section, we will consider the case in which the data queried is electric load consumption \response{forecast} from a specific household.  \response{High quality forecast of power consumption is important data used for resource planning purposes by power system operators and is also desired by third party data analysts that want to provide additional services. However, power consumption data has shown to reveal user behaviors \cite{zoha2012non} where an adversary can infer the presence or absence of household members based on their heat-pump based consumption and electric-vehicle charging patterns that can be extracted from aggregate power consumption. Therefore, providing power consumption forecast can constitute of significant privacy risk and the households do not want to reveal their power consumption forecast lest they risk being identified. On the other hand, power consumption is highly correlated for households that are geographically proximate. To
mitigate the privacy risk, a solution for the households is to obfuscate their power consumption
forecast output while considering households in a certain location and development or are similar dwelling units. In this way, not only
is the privacy risk alleviated, but the utility of the forecasts is also good which is important to the data analyst.}
\par \response{Daily patterns of power consumption are widely available now to utilities as a part of AMI data. They do not fit a Gaussian multivariate distribution directly. }
In prior work \cite{ravi2022differentially}, however, we showed that they fit well a multivariate log-Normal distribution, which implies that the logarithm of the daily pattern is a multivariate Gaussian vector, which is what we need to apply our method.  Let $p[k]$ be the power consumption at hour $k$ during the day.  The model we adopt is as follows:
\begin{equation}
    x[k]=\log (p[k])-\mu[k]
\end{equation}
where $\mu[k]$ is the seasonal mean estimated averaging $\log(p[k])$ over each day so that $x[k]$ is zero mean\footnote{This is not exactly the case since we are dealing with estimates but is our modeling assumption.}. The assumption we make is that $x[k]$ is an ARMA zero mean stationary Gaussian process, which allows us to apply the scheme we proposed in the previous sections. Note that here the query is the forecast of the samples $p[k],~k=K,\ldots,K+T$ which are not Gaussian, i.e.:
\begin{equation}
\begin{aligned}\label{query1}
      \bm q&=[
    \hat p[K],\dots,
    \hat p[K+T-1]]^\top\\
    \hat p[k]&=e^{\hat x[k]+\mu[k]},~~~k=K,\ldots,K+T-1.
\end{aligned}
\end{equation}
Nonetheless the DP answer $\tilde{\bm q}$ can be computed by applying the optimum Gaussian noise to the samples of the process $\hat x[k]+\mu[k], k=K,\ldots,K+T-1$ and then applying the exponential function to them. 
Since the exponential function is a bijective function, the PDP $(\epsilon,\delta)$ privacy guarantee is maintained after this type of post-processing. 
What we can no longer guarantee is that the accuracy of $\tilde{\bm q}$ is $\rho^{\text{MSE}}_{\mathbbm Q|X}$ in \cref{eq:MS-accuracy-constr}. The resulting accuracy can be calculated, but we skip the derivation for brevity. 

\begin{figure*}[!htb]
  \centering
  \includegraphics[width=\textwidth]{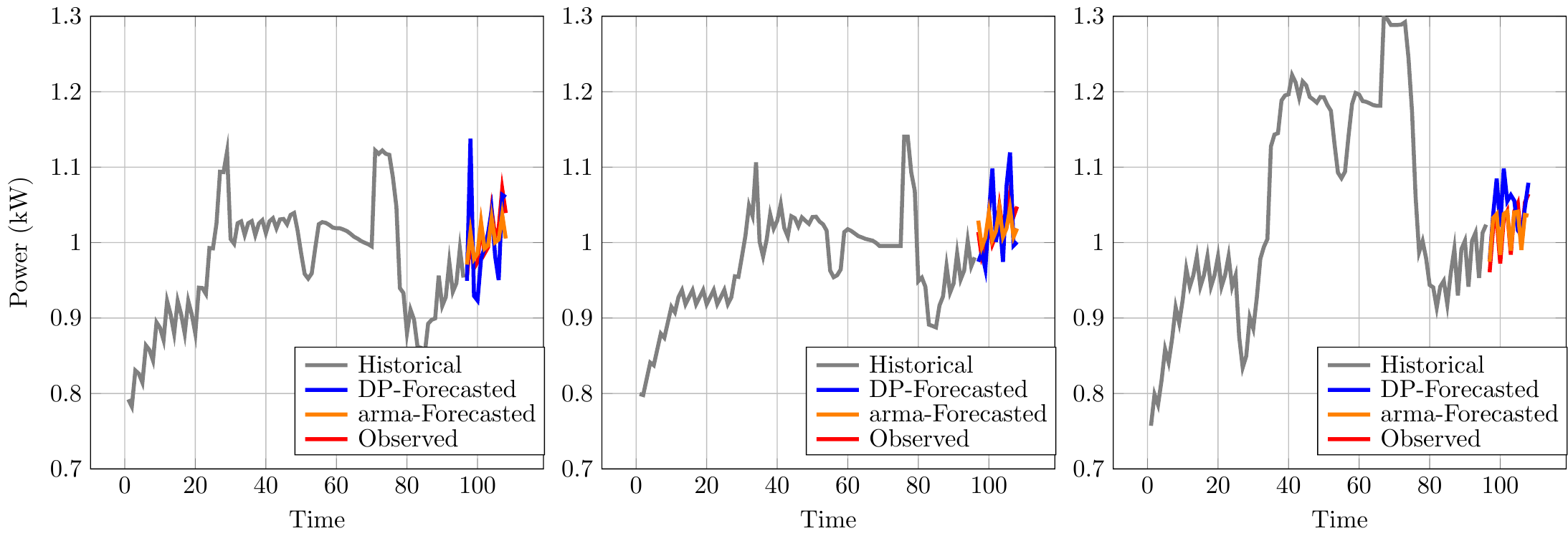}
  	  \caption{ARMA forecasts of power demands of three houses, where the accuracy  of DP-Forecasting is 0.45, i.e., $\rho=0.45$ and $(\epsilon, \delta) = (0.8, 0.5637) $. }
  \label{arma_pred}
\end{figure*}


\begin{figure*}[!htb]
 \centering
 \includegraphics[width=0.961\textwidth]{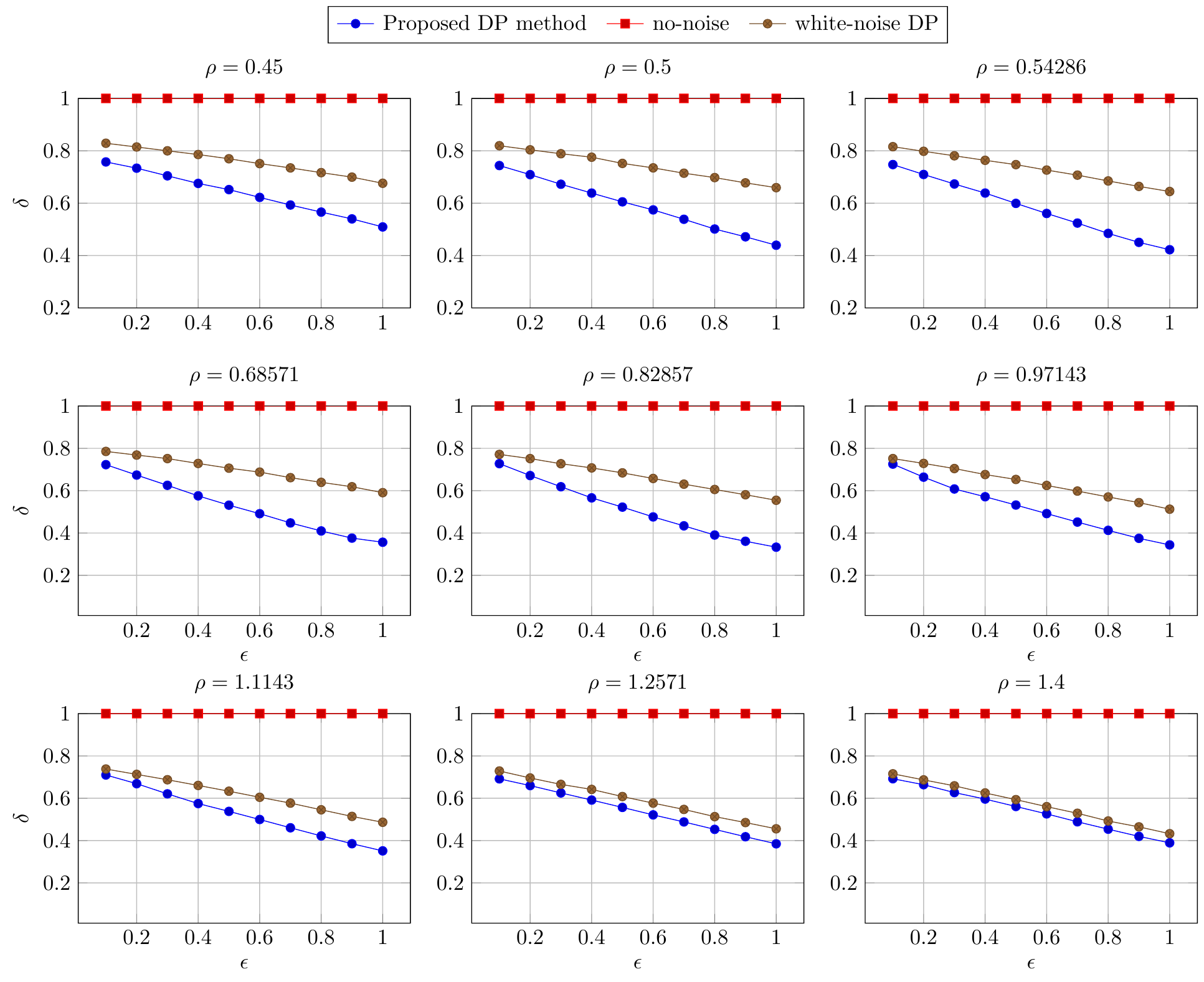}
\caption{The $(\epsilon,\delta)$-private curves with different accuracy budgets $\rho$ for the real-world \response{power consumption} data.}  
   \label{arma_ami}
\end{figure*}

\section{Numerical Results}\label{sec:num_results}
In this section, we illustrate the proposed algorithm with a synthetic dataset and real-world power system \response{power consumption} data.  In the synthetic case, there are several classes and the query corresponding to each class is  bi-dimensional Gaussian with a mean vector and $2\times 2$ covariance  matrix. It is assumed that each class is the neighbor of the other class. \response{Such an assumption is made to compare the results of the proposed method for an example where even adding white Gaussian noise (a baseline) would have good possibility to provide privacy guarantee for a certain level of utility. Note that our proposed method is applicable even in scenarios where this is not the case.} The goal is to conceal the class-label by adding DP noise to the query output. 
\par In the real-world \response{power consumption} case, we first map the forecast of this seasonal and non-Gaussian process onto an ARMA forecasting problem \response{as discussed in Section \ref{subsec:arma_power}}, and then  utilize the $(m=6, n=5)$ ARMA model for training and testing.  
Finally, we utilize the proposed \response{query response mechanism} to protect the ARMA forecasts from households that belong to the same cluster. \response{Here, the idea is that households in the same cluster correspond
to the same geographical region which makes their electricity consumption pattern very geography specific. By considering  all the households in a cluster to be neighbors, utility of the forecasts can be significantly improved while providing good privacy guarantees.  }

\subsection{Synthetic Data}
The synthetic data is generated assuming that there are four classes and all classes are considered as neighbors of the others. The query corresponding to each class follows a two-dimensional Gaussian distribution with mean and covariance matrix. The mean vectors per-class are chosen to lie on a line in the $2$-D space, and the covariance matrix is drawn from a Wishart distribution with scale matrix as the identity and two degrees of freedom. This is shown as a contour plot in \Cref{fig:contour_plot}. Such a case was chosen to highlight the advantage of the proposed method over adding white noise with variance to meet the accuracy requirement.  
\par We compare our method with two baselines. Firstly, the ``no-noise" baseline to understand DP guarantee if no noise was added to the query. One could still expect some inherent privacy, solely due to the fact that the query itself is stochastic and the distributions corresponding to each class-label could be close. Our method starts at this baseline and further improves the privacy guarantee.   
The second baseline is the addition of white Gaussian noise to synthetic data with zero-mean  denoted by ``white noise DP" where the variance of noise $\sigma^2 = \rho/D$, so that the accuracy requirement in \cref{eq:MS-accuracy-constr} is met. Note that the addition of zero-mean white Gaussian noise is commonplace in DP literature, with more focus on choosing the appropriate variance $\sigma^2$. One can think of the baseline as setting of variance in order to meet the accuracy requirement. 
\par We plot the simulation results in \Cref{synthetic_arma} with varying levels of accuracy $\rho$. Note that the no-noise baseline does not depend on the accuracy, and is therefore the same curve throughout. The proposed method outperforms the white-noise baseline irrespective of the accuracy level, although the improvement decreases slightly with larger $\rho$ values. The privacy-accuracy (utility) trade-off is visible in the curves where larger $\rho$ values are associated with better privacy $(\epsilon,\delta)$ values. 



\subsection{Power consumption data forecasting by ARMA Model}
\subsubsection{ARMA Forecasting}
First, we collect 100 households' demands, and then remove the seasonal values (in a weekly scale) from the original demands, and further  take the log values of the nonseasonal demands for training, as described in \cref{arma_process}. Then, we utilize the $k$-means algorithm to cluster the $100$ household's demands into six clusters.  In each cluster, we consider each household   as the neighbor of all other homes. 
Specifically, we choose the 15-day data for training to estimate $a_i$ and $b_j$ parameters in \cref{arma}, where the 15-day data have 1440 samples since the sampling rate  is once per 15 minutes. The forecasting time is three hours with 12 samples, i.e., the dimension of the query is $12$. \response{The forecast trajectory from ARMA model (query), the released forecast trajectory (query response) and the actual power consumption of many households are shown in \Cref{arma_pred}. The gray lines show  the past 24 hours' of data, the red lines show the ground-truth three hours' data, the orange lines show the original query or ARMA forecast trajectory and the blue lines show the released forecast trajectory i.e. query response. }
\response{The first observation  from \Cref{arma_pred} that the ARMA model forecasts (in orange) the future demands with high accuracy compared to , which provides opportunities  for analysts to infer the true demands. This facilitates the need for releasing high-accuracy forecasts with privacy guarantees. The next observation is that applying the proposed query release mechanism (in blue) does not sacrifice the accuracy of forecasts significantly. In \Cref{arma_pred},   the normalized MSE of the released forecast of House 1, House 2 and House 3 are $1.5953e^{-3}$, $1.2475e^{-3}$ and $1.2453e^{-3}$, respectively.}

\subsubsection{Differential Privacy for ARMA Forecasting}
As in the synthetic data case, we show the simulation results with different $\rho$ in \Cref{arma_ami}. Without noise, $\delta$ is always 1, which means that the probability of ARMA forecasted demand being associated with a certain household (class-label) is 100\%. Therefore, it is necessary to design the DP noise for the query of ARMA forecasting. 
With the white noise, $\delta$ is reduced to 0.8 when the accuracy $\rho =0.68571$ and $\epsilon = 0.1$. With $\rho$ fixed and $\epsilon$ increasing from 0.1 to 1.0, $\delta$ decreases slowly from 0.8 to 0.6.
In contrast, with the same $\rho = 0.68571$ and the $\epsilon$ increasing from 0.1 to 1.0, $\delta$ of the proposed optimal DP algorithm decreases fast from 0.72 to 0.35. It indicates that the proposed algorithm has much better performance regarding the privacy protection. Similar to the synthetic data, with $\rho$ increasing (i.e., relaxing the accuracy level), the $(\epsilon, \delta)$ privacy is strengthened.  Another observation is that    with the accuracy budget $\rho$ increasing from 0.45 to 1.4 and with $\epsilon = 1$ fixed, the $\delta$ of the proposed DP algorithm decreases from 0.84 to 0.35, which shows a trade-off between accuracy and privacy. Note that the proposed method offers better guarantee than the white-noise DP method, especially when the accuracy constraint is stronger, i.e., $\rho$ is small. This is quite important considering the utility of the query while also concealing the class-label. 

\section{Conclusion}\label{sec:conclusions}
In this paper, we proposed a new method of differential privacy for label or class-based data, aiming at adding a functional DP noise for the release of query response such that the analyst is unable to infer the underlying class-label. Moreover, we demonstrated that the $(\epsilon, \delta)$-private is guaranteed to be satisfied while meeting the predefined accuracy budget $\rho$.
We illustrate  the effectiveness of the proposed method on the synthetic data, which outperforms the baseline additive white Gaussian noise mechanism.  We further consider the ARMA forecasting problem for the power consumption measurements.  
Then, we implement the proposed DP method to the  ARMA  forecasting method in order to protect the privacy of the \response{power consumption} measurements forecasts. Our empirical case studies on both the synthetic data and  real-world \response{power consumption} measurements validate the effectiveness and advantages  of the proposed method.
\appendix
\subsection{Proof of \texorpdfstring{\Cref{prop:log-likelihood-gauss}}{Proposition 1}}\label{proof:L(xi)}
\begin{proof}
First, note that we can arrange the terms in $ L_{\XX}(\bm q)$ as:
\begin{align}
    L_{\XX}(\bm q)=&\frac{1}{2}\ln | {\bm{\Sigma}}_{\X'} | -  \frac{1}{2} \ln | {\bm{\Sigma}}_{\X} | \nonumber\\
    &+ \frac{1}{2}(\bm{q}-\bm \mu_{\X})\T(\bm \Sigma_{\X'}^{-1}-\bm \Sigma_{\X}^{-1})(\bm{q}-\bm \mu_{\X})\nonumber\\
    &+(\bm \mu_{\X}-\bm \mu_{\X'})\T\bm \Sigma_{\X'}^{-1}(\bm{q}-\bm \mu_{\X})\nonumber\\
    &+\frac{1}{2}(\bm \mu_{\X'}-\bm \mu_{\X})\T\bm \Sigma_{\X'}^{-1}(\bm \mu_{\X'}-\bm \mu_{\X}).\label{eq:LX-Gaussian}
\end{align}
For simplicity, we omit the suffix $XX'$. First, \cref{eq:xi} re-centers, whitens and then rotates the normal vector $\bm q$, which results in  the i.i.d normal vector $\bm \xi$. To derive the expression of $L(\bm \xi)$ from \cref{eq:LX-Gaussian}
we observe that:
\begin{align*}
(\bm \Sigma_{\X'}^{-1}-\bm \Sigma_{\X}^{-1})&=\bm 
\Sigma_{\X}^{-1/2}
(\bm
\Sigma_{\X}^{1/2}\bm \Sigma_{\X'}^{-1}\bm
\Sigma_{\X}^{1/2}-\bm I)\bm
\Sigma_{\X}^{-1/2}\\
&=\bm 
\Sigma_{\X}^{-1/2}\bm U\bm U\T
(\bm
\Sigma_{\X}^{1/2}\bm \Sigma_{X'}^{-1}\bm
\Sigma_{\X}^{1/2}-\bm I)\bm U\bm U\T\bm
\Sigma_{\X}^{-1/2}\\
&=\bm \Sigma_{\X}^{-1/2}\bm U(\bm \Gamma-I)\bm U\T\bm
\Sigma_{\X}^{-1/2}
\end{align*}
which implies that 
\[(\bm{q}-\bm \mu_{\X})\T(\bm \Sigma_{\X'}^{-1}-\bm \Sigma_{\X}^{-1})(\bm{q}-\bm \mu_{\X})=\bm \xi\T(\bm \Gamma-I)\bm \xi.\]
The other two terms are obtained similarly:
\begin{align*}
    (\bm \mu_{\X}\!-\!&\bm \mu_{\X'})\T\bm \Sigma^{-1}_{\X'}(\bm q\!-\!\bm \mu_{\X})=(\bm \mu_{\X}\!-\!\bm \mu_{\X'})\T\bm \Sigma_{\X}^{-1/2}\bm U\\
    &\cdot \bm U\T(\bm 
    \Sigma_{\X}^{1/2}\bm \Sigma^{-1}_{\X'}\bm \Sigma_{\X}^{1/2})\bm U \bm U\T\bm \Sigma_{\X}^{-/2}(\bm q\!-\!\bm \mu_{\X})=-\bm \mu\T \bm \Gamma\bm \xi.
\end{align*}
\begin{align*}
    (\bm \mu_{\X'}\!-\!&\bm \mu_{\X})\T\bm \Sigma^{-1}_{\X'}(\bm \mu_{\X'}\!-\!\bm \mu_{\X})=(\bm \mu_{\X'}\!-\!\bm \mu_{\X})\T\bm \Sigma_{\X}^{-1/2}\bm U\\
    &\cdot \bm U\T(\bm 
    \Sigma_{\X}^{1/2}\bm \Sigma^{-1}_{\X'}\bm \Sigma_{\X}^{1/2})\bm U
    \bm U\T\bm \Sigma_{\X}^{-1/2}(\bm \mu_{\X'}\!-\!\bm \mu_{\X})=
    \bm \mu\T \bm \Gamma\bm \mu.
\end{align*}
\end{proof}
\subsection{Proof of \texorpdfstring{\Cref{lem:chernoff-bound}}{Lemma 1}}\label{app:chernoff-bound}

\begin{proof}
Again, we omit the suffixes $XX'$ to streamline the notation. The idea is to leverage \Cref{prop:log-likelihood-gauss} and find a bound for $Pr(L(\bm \xi)>\epsilon)$ through the Chernoff bound. More specifically, considering the expression \cref{eq:L(xi)}, for $\overline{s}>0$:
\begin{align*}
Pr(L(\bm \xi)>\epsilon)&=
Pr(e^{\frac {\overline{s}} 2 \left(\bm \xi\T(\bm \Gamma\!-\!\bm I)\bm \xi-2{\bm \mu}\bm \Gamma \bm \xi+{\bm \mu}\T\bm \Gamma {\bm \mu}\right)} >e^{\overline{s}\epsilon+\frac {\overline{s}} 2\ln|\bm \Gamma|})\\
&\leq \frac{
\mathbb{E}\left[
e^{\frac {\overline{s}} 2 \left(\bm \xi\T(\bm \Gamma\!-\!\bm I)\bm \xi-2{\bm \mu}\bm \Gamma \bm \xi+{\bm \mu}\T\bm \Gamma {\bm \mu}\right)}
\right]
}{
e^{\overline{s}\epsilon}|\bm \Gamma|^{\overline{s}/2}} 
\end{align*}
The expectation $\mathfrak{J}(\bm \Gamma, {\bm \mu})\triangleq{\mathbb E}[e^{
   \frac{\overline{s}}{2}
    \left( \bm \xi\T(\bm \Gamma\!-\!\bm I)\bm \xi-2{\bm \mu}\T\bm \Gamma\bm \xi+ {\bm \mu}\T\bm \Gamma{\bm \mu}\right)
   }]$ can be evaluated in closed form by solving the following integral:
\begin{align}\label{eq:integral}
\mathfrak{J}(\bm \Gamma, {\bm \mu})=
\!\int \!\frac{
   e^{
   \frac{\overline{s}}{2}
    \left( \bm \xi\T(\bm \Gamma\!-\!\bm I)\bm \xi-2{\bm \mu}\T\bm \Gamma\bm \xi+ {\bm \mu}\T\bm \Gamma{\bm \mu}\right)-\frac{\bm \xi\T\bm \xi}{2}
   }}{(2\pi)^{k/2}}~\mathrm{d}\bm \xi.
\end{align}
The goal is to express the exponent as a quadratic form. Combining the quadratic terms and extracting the factor $ -\frac{\overline{s}}{2}$ the expression becomes:
\begin{align*}
    -\frac{\overline{s}}{2}
    \left(
    \bm{\xi}\T\left((1\!+\!\overline{s}^{-1})\bm I\!-\!\bm \Gamma\right)\bm\xi
    +2{\bm \mu}\T\bm \Gamma \bm\xi-{\bm \mu}\T\bm \Gamma{\bm \mu}
    \right)
\end{align*}
The expression is streamlined by the following substitution for $\overline{s}$ and definition of
$\hat{\bm \mu}$:
\begin{align}
    s&=1\!+\overline{s}^{-1}~\Rightarrow~\overline{s}=(s-1)^{-1}, \overline{s}>0 \implies s > 1\\
    \hat{\bm \mu}&\triangleq -(s\bm I-\bm\Gamma)^{-1}\bm \Gamma{\bm \mu}.
\end{align}
so that the exponent can then be rearranged as follows:
\begin{align*}
   -&\frac{
\bm \xi\T(s\bm I\!-\!\bm \Gamma)\bm \xi-2\hat{\bm\mu}\T (s\bm I\!-\!\bm \Gamma)\bm\xi
\pm\hat{\bm\mu}\T (s\bm I\!-\!\bm \Gamma)\hat{\bm\mu}
-{\bm \mu}\T\bm \Gamma{\bm \mu}}{2(s\!-\!1)}=\\
-&\frac{(\bm \xi-\hat{\bm \mu})\T(s\bm I\!-\!\bm \Gamma)
(\bm \xi-\hat{\bm \mu})}{2(s\!-\!1)}+\frac{s}{2(s\!-\!1)}{\bm \mu}\T(s\bm I\!-\!\bm \Gamma)^{-1}\bm \Gamma{\bm \mu}
\end{align*}
hence, if and only if $s\bm I -\bm \Gamma$ is positive definite, which means that $s>\gamma_1=\lambda_{\max}(\bm\Sigma_{\X}^{1/2}\bm \Sigma_{\X'}^{-1}\bm\Sigma_{\X}^{1/2})$, we have:
\begin{align*}
    \mathfrak{J}(\bm \Gamma, {\bm \mu})&=\frac{e^{\frac{s}{2(s\!-\!1)}{\bm \mu}\T(s\bm I\!-\!\bm \Gamma)^{-1}\bm \Gamma{\bm \mu}}}{(2\pi)^k}\int e^{-\frac{(\bm \xi-\hat{\bm \mu})\T(s\bm I\!-\!\bm \Gamma)
(\bm \xi-\hat{\bm \mu})}{2(s\!-\!1)}}~\mathrm{d}\bm \xi
\\&=
\frac{e^{\frac{s}{2(s\!-\!1)}{\bm \mu}\T(s\bm I\!-\!\bm \Gamma)^{-1}\bm \Gamma{\bm \mu}}(s\!-\!1)^{k/2}}
{|s\bm I\!-\!\bm \Gamma|^{1/2}}
\end{align*}
We can conclude that:
\begin{align}
    Pr(L(\bm \xi)>\epsilon)&\leq 
    \frac{(s-1)^{k/2}e^{-\frac{\epsilon}{(s\!-\!1)}+\frac{s}{2(s\!-\!1)}{\bm \mu}\T(s\bm I-\bm \Gamma)^{-1}\bm \Gamma{\bm \mu}}}
{|\bm \Gamma|^{\frac{1}{2(s\!-\!1)}}|s\bm I-\bm \Gamma|^{1/2}}
\end{align}
where the bound is the expression in \cref{eq:chernoff}.
\end{proof}
\subsection{Proof of \texorpdfstring{\Cref{lem:alpha_t}}{Lemma 2}}\label{app:step_size_t}

\begin{proof}
	Upon updating  $\bm{A}_{\X_{\bullet}(t) } (t)$, the cost $J_{t+1}$ is, 
	\begin{align}
	J_{t+1} = \sup_{\X   } \sup_{  \X' \in  {\cal X}_X^{(1)}   }  \left\{ g_{\X_{\bullet}{(t)} \X_{\bullet}(t)},  g_{\X  \X_{\bullet}(t)    }, g_{\X_{\bullet}(t)\X}, g_{\XX} \right\}  
	\end{align}
	We now divide the analysis into different cases to derive the condition on the  step size. 
	\begin{itemize}
     \item 	If $J_{t+1} = g_{\XX}$  where $\X \neq \X_{\bullet}(t), \X' \neq \X_{\bullet}(t) $, then $J_{t} - J_{t+1} \geq 0 $ since
	 terms $g_{\XX}$ were unaffected by the update and at step $t$ $g_{\X_{\bullet}{(t)} \X_{\bullet}(t)}$ was the supremum thus,  
	 $J_{t} - J_{t+1}  = g_{\X_{\bullet}{(t)} \X_{\bullet}(t)} -g_{\XX} \geq 0 $. 
	 \item If $J_{t+1} = g_{\X_{\bullet}{(t)} \X_{\bullet}(t)}$, it means that the chosen term to minimize is still the supremum,
	 \begin{align}
	 J_{t} - J_{t+1} &= g_{\X_{\bullet}{(t)} \X_{\bullet}(t)} ( \bm{A}_{\X_{\bullet}(t) } (t), \bm{A}_{\X_{\bullet}(t) } (t)  ) \nonumber\\
	 &- g_{\X_{\bullet}{(t)} \X_{\bullet}(t)} ( \bm{A}_{\X_{\bullet}(t) } (t+1),\bm{A}_{\X_{\bullet}(t) } (t)  ) 
	 \end{align}
	 Since $g_{\X_{\bullet}{(t)} \X_{\bullet}(t)} ( \bm{A}_{\X_{\bullet}(t) } (t), \bm{A}_{\X_{\bullet}(t) } (t)  ) $ is convex in $\bm{A}_{\X_{\bullet}(t) } (t)$, the update for $\bm{A}_{\X_{\bullet}(t) } (t)$ ensures that $J_{t} - J_{t+1} > 0$. However, for an appropriate step-size, one can either resort to backtracking line search or use a small fixed step size. 
	 \item If $J_{t+1} =  g_{\X  \X_{\bullet}(t)    }$, to derive the condition for $\alpha_t$ such that $J_{t} - J_{t+1} >0$ consider,
	 \begin{align}
	       J_{t} - J_{t+1} &= g_{\X_{\bullet}{(t)} \X_{\bullet}(t)}\left( \bm{A}_{\X_{\bullet}(t) } (t), \bm{A}_{\X_{\bullet}(t) } (t)  \right) \\
	       &-  g_{\X  \X_{\bullet}(t)    } \left(  \bm{A}_{\X_{\bullet}(t) } (t+1), \bm{A}_{\X } (t)   \right) \\
	       & = \tilde{\bm{v}}_{\X_{\bullet}\X_{\bullet} }^{\T} \bm{A}_{\X_{\bullet} }^{ (t)} \tilde{\bm{v}}_{\X_{\bullet}\X_{\bullet} }  - \tilde{\bm{v}}_{\X \X_{\bullet} }^{\T} \bm{A}_{\X_{\bullet}(t) } (t+1) \tilde{\bm{v}}_{\X \X_{\bullet} } \nonumber\\
	       & - \ln | \bm{A}_{\X_{\bullet}(t)} (t)  | + \ln | \bm{A}_{\X_{\bullet} (t)} (t)  | \nonumber \\
	       &+ \ln | \bm{A}_{\X_{\bullet}(t)} (t+1)  | - \ln | \bm{A}_{\X} (t)  |
	 \end{align}
	 With further simplification and by imposing the condition $J_{t} - J_{t+1} >0$ , we get 
	 \begin{align}
	    \alpha_{\X_{\bullet}(t) }\leq  d_{\X_{\bullet}(t), \X } 
	 \end{align}
	 \item If $J_{t+1} =  g_{\X_{\bullet}(t)\X}$, consider $J_{t} - J_{t+1} $, 
	 \begin{align}
	    J_{t} - J_{t+1} &= g_{\X_{\bullet}{(t)} \X_{\bullet}(t)}\left( \bm{A}_{\X_{\bullet}(t) } (t), \bm{A}_{\X_{\bullet}(t) } (t)  \right) \nonumber\\
	    &-  g_{  \X_{\bullet}(t) \X    } \left(  \bm{A}_{\X  } (t), \bm{A}_{\X_{\bullet} (t) } (t+1)   \right) \\
	    & =  \tilde{\bm{v}}_{\X_{\bullet}\X_{\bullet} }^{\T} \bm{A}_{\X_{\bullet} }^{ (t)}   \tilde{\bm{v}}_{\X_{\bullet}\X_{\bullet} } - \tilde{\bm{v}}_{\X \X_{\bullet} }^{\T} \bm{A}_{\X  } (t) \tilde{\bm{v}}_{\X \X_{\bullet} } \nonumber\\
	    &  - \ln | \bm{A}_{\X_{\bullet}(t)} (t)  | + \ln | \bm{A}_{\X_{\bullet} (t)} (t)  | \nonumber \\
	    &+ \ln | \bm{A}_{\X} (t)  | - \ln | \bm{A}_{\X_{\bullet}(t)} (t+1)  |
	 \end{align}
	  With further simplification and by imposing the condition $J_{t} - J_{t+1} >0$ , we get:
	  \begin{align}
	  \alpha_{\X_{\bullet}(t) } \leq  b_{\X_{\bullet}(t), \X } 
	  \end{align}
	  where $b_{\X_{\bullet}(t)}$ is given in \eqref{eq:b_xx}.
	 \end{itemize}
\end{proof}
\bibliographystyle{IEEEtran}
\bibliography{Differential_privacyMarkov}
\end{document}